\documentclass{article}
\usepackage{amsmath,amsthm,amsfonts,graphicx}
\usepackage{multirow}
\usepackage{slashbox}
\usepackage{multirow}
\def\smallddots{\mathinner{\raise7pt\hbox{.}\raise4pt\hbox{.}\raise1pt\hbox{.}}}
\def\smallsdots{\mathinner{\raise1pt\hbox{.}\raise4pt\hbox{.}\raise7pt\hbox{.}}}

\numberwithin{equation}{section}
\numberwithin{table}{section}
\newtheorem{theorem}{Theorem}[section]

\newtheorem{example}{Example}[section]

\usepackage{graphicx}
\graphicspath{%
    {converted_graphics/}
}
\begin{document}

\centerline{{\Large \bf Fast Feasible and Unfeasible Matrix Multiplication}}  

\medskip

\medskip

\centerline{Victor Y. Pan} 

\medskip

\medskip

\centerline{$^{[1]}$~~Department of Computer Science}

\centerline{Lehman College of the City University of New York}

\centerline{Bronx, NY 10468 USA}

\centerline{victor.pan@lehman.cuny.edu}
\centerline{http://comet.lehman.cuny.edu/vpan/}
\centerline{and} 
 
\centerline{$^{[2]}$~~Ph.D. Programs in Mathematics  and Computer Science}
\centerline{The Graduate Center of the City University of New York}
\centerline{New York, NY 10036 USA}

 \date{}



 
\begin{abstract}
Matrix-by-matrix multiplication (hereafter {referred to as \em MM}) is a  fundamental operation omnipresent in modern computations in Numerical and Symbolic Linear Algebra.
Its acceleration makes major impact on various fields of Modern Computation and has been a highly recognized research subject for about five decades. 
The researchers introduced
 amazing novel techniques, found
 new insights into MM and numerous  related  computational problems, and devised advanced algorithms that performed  $n\times n$ MM by using less than $O(n^{2.38})$ scalar arithmetic operations
versus $2n^3-n^2$
of the  straightforward MM, that is, 
more than half-way to the information lower bound $n^2$.  The record upper bound 3 of 1968 on the exponent of the  complexity MM decreased below 2.38 by 1987 and has been extended to  various celebrated  problems in many areas of computing and became most extensively cited constant of the Theory of Computing. The progress in decreasing the record exponent, however, has virtually  stalled since 1987, while many scientists are still anxious to know its sharp bound, so far restricted to the range from 2 to about 2.3728639.
Narrowing this range remains a celebrated challenge. 

Acceleration of MM  in the Practice of Computing is a distinct challenge, 
much less popular, but also highly important. Since 1980 the progress towards meeting the two challenges has followed two distinct paths because of the {\em curse of recursion} -- all the known algorithms 
supporting the exponents below 2.38 or even below 2.7733 involve long sequences of nested recursive steps, which 
 blow up the size of an input matrix.  
 As a result all these algorithms
improve straightforward MM
  only for unfeasible MM of immense size,  greatly exceeding the sizes of interest nowadays and in any foreseeable future. 
  
 It is plausible and surely highly desirable that someone could eventually decrease the record MM exponent towards its sharp bound 2 even without ignoring the curse of recursion, but currently there are two distinct challenges of the acceleration of feasible and unfeasible MM.   
  
   In particular various known algorithms supporting the exponents in the range between 2.77 and 2.81 are quite efficient for feasible MM and have been implemented. Some of them make up a valuable part of modern software for numerical and symbolic matrix computations, extensively worked on in the last decade.  Still, that work has mostly relied on the MM algorithms proposed more than four decades ago, while more efficient algorithms are well known,  some of them appeared in 2017.
   
  In our
 review we first survey the mainstream study of the acceleration of MM of unbounded sizes, cover the progress in decreasing the exponents of MM, comment on its impact  on the theory and practice of computing, and  recall various fundamental concepts and techniques supporting fast MM
and naturally introduced in that study by 1980.  Then we demonstrate how the curse of recursion naturally entered the game
of decreasing the record exponents. Finally we cover the State of the Art of efficient  feasible MM, including some most efficient known techniques and algorithms as well as  various issues of  numerical and symbolic implementation. 

We hope that our review  will help motivate and properly focus further effort in this highly important area.
 
\end{abstract} 







\section{Introduction}\label{sintr}





Matrix-by-matrix multiplication (MM)
is omnipresent in modern computations
for Applied Mathematics,  Statistics, Physics, Engineering,
   Combinatorics, Computer Science, Signal and Image Processing, and Bioinformatics (cf., e.g., \cite{LT91}, \cite{P94})
 and is performed billions times per day around the globe.
    For a couple of specific examples, computing radiation exchange in enclosures can be written in
matrix form, mainly as an $N\times N$ MM with $N$ up to 10,000 \cite{GEKM14}, and gene expression analysis
in bioinformatics involves $N\times N$ MMs 
with $N$ up to 4,000 \cite{S12}.


The  straightforward  
MM involves cubic time, namely
$2n^3-n^2$ arithmetic operations for $n\times n$, and  until 1969
 it was commonly believed that 
 the order of $n^3$ operations are necessary. This belief died in 1969, when Volker Strassen decreased the exponent 3 of the complexity of MM to 2.8074 in \cite{S69}. Further decrease of the exponent towards  its information lower bound 2 has instantly become the subject of  worldwide interest and the goal of intensive study by literally all experts in this field. 
Various novel amazing but fundamental techniques have been proposed, new insights into fast MM and some related areas of modern computations have been introduced, and the MM exponent has been decreased below 2.38, that is, more than half-way from the classical 3 to the information lower bound 2. These results have been  widely recognized in the Theory of Computation, and
2.38 became the absolute constant of that field, most frequently cited there,  because it bounds the exponents of the complexity of various celebrated  problems in many areas of computing linked to MM. 

The progress was not going smoothly. It  stalled for nearly a decade after 1969, resumed in 1978 with \cite{P78}, followed with new significant advances in 1979--1982 and again in 1986, and since then again virtually  stalled.
Even more serious problem with the record-breaking MM algorithms is {\em the curse of recursion} -- they rely on long sequences of nested recursive processes, which 
 blow up the input size. The resulting algorithms supersede the straightforward MM  only when this size becomes immense -- greatly exceeding the level of MM in use nowadays or even in any foreseeable future.
 
 It is plausible and of course highly desirable that new breakthrough 
 would eventually result in an algorithm for fast feasible MM supporting a nearly sharp exponent close to 2, but 
so far the acceleration of  feasible MM of realistic sizes is studied as a distinct hard task for which one cannot ignore  the curse of recursion. So far the known algorithms for feasible MM only support the exponents below 2.7734
 (see \cite{P82}), unbeaten for 36 years since 1982 but not much recognized.  
Nevertheless in spite of their  association with 
relatively large exponents,
some fast algorithms for feasible MM make up a valuable part of modern software for numerical and symbolic matrix computations.
 Implementation work for fast feasible MM has been intensified within the last decade but still mostly relies on 
 algorithms that are more than four decades old and have already been significantly improved (see
\cite{K04}, \cite{S13}, and \cite{KS17}).
  
In our present survey 
 we first trace the history of the  decrease of the  exponent of MM,
 list some celebrated problems of the Theory of Computing
linked to MM and sharing with MM the exponent of the complexity, and naturally recall the fundamental concepts and techniques for fast MM, introduced by 1981, such as {\em recursive  bilinear and trilinear algorithms}, {\em tensor decomposition}, {\em trilinear aggregation},  {\em disjoint MM},   {\em any precision approximation}, 
 and the {\em EXPAND, PRUNE, and CONQUER}
 techniques. We also  show how the curse of recursion naturally entered the game. Finally we 
  bring our
survey closer to the earth by covering the 
acceleration of feasible MM of  realistic sizes and some relevant methods and techniques
such as the {\em 3M method} for complex MM,
its extension to {\em polynomial MM}, 
{\em computer-aided ALS search}, and {\em randomization technique} for fast MM.  We also discuss  
 the  implementation issues of fast MM, including numerical stability, data movement, and symbolic implementation. 

 We hope that our work will help motivate and properly focus further study 
and further progress in  this highly important area of Modern Computations.

 

\section{Straightforward MM and First Fast MM Algorithms}\label{sblnpa} 


\subsection{MM and its subproblems MV and VV. Optimality of  straightforward MV and VV}\label{sstrvvmv} 


 $k\times m$ by $m\times n$ MM 
 (hereafter referred to as MM$(k,m,n)$)
 is the computation of the product $C$ of two matrices $A$ of size $k\times m$ and
$B$ of size $m\times n$ whose entries can be numbers, variables, or matrices:
\begin{equation}\label{eqMM}
C=AB,~{\rm for}~
A=(a_{i,j})_{i,j=1}^{k,m},~B=(b_{j,h})_{j,h=1}^{m,n},
~C=(c_{{i,h}})_{i,h=1}^{k,n}, 
\end{equation}
$$c_{i,h}=a_{i,1}b_{1,h}+a_{i,2}b_{2,h}+\cdots+a_{i,m}b_{m,h},~~i=1,\dots,k,~~
h=1,\dots,n.$$
 
MM$(k,m,1)$
is {\em matrix-by-vector multiplication} (hereafter  referred to as {\em  MV)},
$$c_{i}=a_{i,1}b_{1}+a_{i,2}b_{2}+\cdots+a_{i,m}b_{m},~~i=1,2,\dots,k;$$
$M(1,m,1)$ is the computation of the {\em inner product of two vectors} of length $m$ (hereafter 
{\em VV}),
$$c=a_{1}b_{1}+a_{2}b_{2}+\cdots+a_{m}b_{m},$$
and we write MM$(n)$ for MM$(n,n,n)$.

 MV and MM$(n)$ are made up of their $n$ and $kn$ subproblems VV, respectively. Like MM both MV and VV are omnipresent in modern computations in Linear and Multilinear Algebra. 

The straightforward classical algorithm for  VV first computes the $m$ products  
$a_1b_1$, $a_2b_2,\dots,$ $a_mb_m$ and then sum them together. For 
MV and MM we can apply the same algorithm $k$ and $kn$ times, respectively. This  is optimal
for VV and MV --  any algorithm that performs VV or MV by using 
only scalar arithmetic operations, that is, {\em add, subtract, multiply, and divide}, and
that uses no branching must perform 
at least as many scalar additions and subtractions and at least as 
many scalar multiplications and 
divisions as the straightforward algorithm.  The proofs rely 
  on the techniques of {\em active operation -- basic substitution}
from \cite{P66},  also covered and extended in \cite[Section 2.3]{BM75}, \cite{K97}, and  Sections  ``Pan's method" in \cite{S72} and \cite{S74}.

Sparse and structured MV, also omnipresent in modern computations,
  can be performed much faster. In particular   arithmetic time  linear  in $n$ up to a logarithmic or polylogarithmic factor is sufficient in order to 
 multiply  by a vector an  
$n\times n$ matrix with a structure of Toeplitz, Hankel, Cauchy, or Van\-der\-monde type and even to solve a nonsingular structured linear system of $n$ equations with such a coefficient
matrix  
  versus quadratic or cubic time 
 required for the same computations with a general matrix. See \cite{P01}, \cite{P15},
and the bibliography therein for computations with structured matrices
(aka {\em data sparse} matrices);  see
\cite{BG12}, \cite{BDKSa},
 \cite{ABBDa}, and the bibliography therein for computations with sparse matrices.


\subsection{The first fast MM algorithms}\label{sffmm} 

 
The  straightforward MM uses 
$kmn$ scalar multiplications and
$kmn-kn$ scalar additions, that is, 
 $n^3$ and
$n^3-n^2$ for $k=m=n$. 
Until 1969 the scientific world 
believed that this  algorithm is optimal, although already in  1968 the experts  knew   from the following example that this 
was not true (see \cite{W68}
and notice technical similarity to the algorithms for polynomial evaluation with 
preprocessing in \cite{P66} and \cite{K97}).

\begin{example}\label{exw0}  
For any even positive integer $n$, the inner product of two vectors  ${\bf a}=(a_j)_{j=1}^n$
and ${\bf b}=(b_j)_{j=1}^n$  satisfies the following identity,
\begin{equation}\label {eqcmmm}
{\bf a}^T{\bf b}=\sum_{i=1}^{n/2}(a_{2i-1}+b_{2i})(b_{2i-1}+a_{2i})-
\sum_{i=1}^{n/2}a_{2i-1}a_{2i}-\sum_{i=1}^{n/2}b_{2i-1}b_{2i}.
\end{equation} 
 By combining such 
 identities for  $n^2$ inner products 
defining $n\times n$ MM, perform  
MM$(n)$ by using  $0.5  n^3 +n^2$ 
scalar multiplications
and $1.5n^3+2n^2-2n$ scalar additions and subtractions, thus replacing 
about 50\% of multiplications of the straightforward MM by additions.
\end{example} 
 
In the 1960s a floating point multiplication was usually two or three times slower than a  floating point addition,
and so the algorithm had some practical value. Not so anymore nowadays because multiplication is about as fast as addition, but the algorithm is still a historical landmark as the {\em first fast MM}. It was not fast enough in order to attract the attention of non-experts, however, and so 
 the news that the straightforward  MM is not optimal has awaken scientific world only a little later, 
when Volker Strassen  presented the following celebrated algorithm. 

\begin{example}\label{ex0} 
{\em Strassen's $2\times 2$ MM}  \cite{S69}. \\
Compute the product $C=AB$ of a pair of $2\times 2$ matrices, for
\begin{equation}\label{eqmmABC}
A=\begin{pmatrix} 
a_{11} & a_{12} \\ 
a_{21} & a_{22}
 \end{pmatrix},~B=\begin{pmatrix} 
b_{11} & b_{12} \\ 
b_{21} & b_{22}
 \end{pmatrix},~C=AB=\begin{pmatrix} 
c_{11} & c_{12} \\ 
c_{21} & c_{22}
 \end{pmatrix},
 \end{equation}
  by using the following expressions, 
 
 $p_1=(a_{11}+a_{22})(b_{11}+b_{22}),~p_2=(a_{21}+a_{22})b_{11},~p_3=a_{11}(b_{12}-b_{22})$,  
$p_4=a_{22}(b_{21}-b_{11}), \\
~p_5=(a_{11}+a_{12})b_{22},~
p_6=(a_{21}-a_{11})(b_{11}+b_{12})$,  
$p_7=(a_{12}-a_{22})(b_{21}+b_{22})$,  \\
$c_{11}=p_1+p_4+p_7-p_5,~c_{12}=p_3+p_5,~c_{21}=p_2+p_4,~c_{22}=p_1+p_3+p_6-p_2$.
\end{example}

The algorithm uses 7 scalar multiplications and 18 scalar
additions and subtractions versus 8 scalar multiplications and 4 scalar
additions of the  straightforward MM.
The trade-off seems to favor the straightforward algorithm, but let all matrix entries 
be  $2\times 2$ matrices,  reapply the algorithm for all 7 auxiliary $2\times 2$ MMs,
and arrive at   
$4\times 4$ MM  using 49 scalar  multiplications versus 64 in the straightforward MM. Then replace the input entries of the new algorithm by $2\times 2$ matrices and reapply the algorithm of  Example \ref{ex0} to all 49 auxiliary $2\times 2$ MMs.
Perform $d$ such recursive steps and arrive at $n\times n$  MM that uses $7^g=n^{\log_2 (7)}$  scalar  multiplications for $n=2^g$ versus $8^g=n^3$ in the  straightforward algorithm.
Addition or subtraction of $s\times s$ matrix 
involves just $s^2$ scalar additions or subtractions, and  the recursive algorithm based on Example \ref{ex0}  performs MM$(n)$
for $n=2^g$ by applying
$6\cdot 7^g-6\cdot 4^g<6n^{\log_2 (7)}$ scalar additions and subtractions.
This is an example of 
 the large class of divide and conquer
 algorithms, which recursively reduce an 
 original  computational problem to those of smaller size. Such algorithms have been extensively used  in various areas of computing, fast Fourier transform being a notable example
 (see Appendix \ref{scnvfft}).

One can decrease the arithmetic cost   
to $5\cdot 7^g-5\cdot 4^g<5n^{\log_2 (7)}$
by applying similar recursive divide and conquer process based on the algorithm of the following example, which performs $2\times 2$ MM by
using 7 scalar multiplications and 15 scalar additions and subtractions. 

\begin{example}\label{ex1} {\em Winograd's $2\times 2$ MM}
(cf.\  \cite{F74}, \cite[pages 45--46]{BM75}).

Compute a $2\times 2$ matrix product $C=AB$ of (\ref{eqmmABC}) 
 by using the following expressions, 

$s_1=a_{21}+a_{22}$, 
$s_2=s_1-a_{11}$,
$s_3=a_{11}-a_{21}$,
$s_4=a_{12}-s_2$,

$s_5=b_{12}-b_{11}$,
$s_6=b_{22}-s_5$, 
$s_7=b_{22}-b_{12}$,
$s_8=s_6-b_{21}$, 

$p_1=s_2s_6,~p_2=a_{11}b_{11},~p_3=a_{12}b_{21}$,  \\ 
	$p_4=s_3s_7,~p_5=s_1s_5,~p_6=s_4b_{22}$, 
$p_7=a_{22}s_8$,  

$t_1=p_1+p_2,~t_{2}=t_1+p_4$, $t_{3}=t_1+p_5$,  \\
$c_{11}=p_2+p_3,~c_{12}=t_3+p_6,~c_{21}=t_2-p_7,~c_{22}=t_2+p_5$.
\end{example} 

With some additional care one can perform $n\times n$ MM for any $n$ by using 
$c\cdot n^{\log_2 (7)}$ arithmetic operations
for $c\approx 4.54$ and $c\approx 3.92$
based on  Examples \ref{ex0}
and  \ref{ex1}, respectively (see \cite{F74}).  
 

\subsection{Impact of  the acceleration of MM}\label{sblnpex} 

 
 The news that the cubic arithmetic time of the straightforward MM is not a barrier anymore flew many times around the globe as a scientific sensation of the year of 1969, and the scientists expected that very soon the classical exponent 3
 will be decreased to its information lower bound 2,
 that is, that very soon $n\times n$ MM will be performed in quadratic or nearly 
 quadratic time,
  required already in order to access the 
 $2n^2$ input entries as well as in order to output $n^2$ entries of the matrix.
 Of course this decrease was a most exciting  perspective --
 according to \cite[page 248]{TB97} ``such a development would trigger the greatest upheaval in the history of numerical computations."
 
 The scientists in many fields were excited 
 because it was known that the acceleration of MM can be readily extended to a 
 variety of popular and long-studied problems of Linear   
Algebra and Computer Science, linked to MM and
 sharing with MM the exponent of their complexity estimates.\footnote{The algorithm designers
 try hard to reduce various problems of modern numerical and symbolic computations to MM, with no considerable overhead, because MM, and even the straightforward cubic time MM, has been very efficiently  implemented on  modern computers.} Their 
 list includes
Boolean MM, computation of paths and distances in graphs,  parsing con\-text-free grammars, 
the solution of a nonsingular linear system of equations, computations of the inverse, determinant,
characteristic and minimal polynomials, and various factorizations of a matrix (see
  \cite{S69}, \cite{BH74}, 
 \cite[Sections 6.3--6.6]{AHU74}, \cite[pages 49--51]{BM75},
\cite{V75}, \cite[Sections 18--20]{P84b},
 \cite[Chapter 2]{BP94}, \cite{AGM97}, \cite{DI00},  \cite{L02}, \cite{Z02}, 
 \cite{YZ04}, \cite{YZ05}, \cite{YZ05a},  \cite{KSV06}, \cite{BJS08}, 
 \cite{AP09}, \cite{DP09}, \cite{Y09}, \cite{SM10}, 
 \cite{LG12}, 
 and \cite{HLSa}. 
 In particular  Strassen's acceleration of MM in \cite{S69} implied  the decrease of the known complexity exponent 3 of the cubic solution time  to $\log_2(7)\approx 2.8074$ for all these computational problems.
 
 The challenge brought  MM and the complexity of  algebraic computations to the  limelight and motivated 
tremendous effort of numerous researchers around the globe, who competed for breaking the record of \cite{S69}.  As Donald E. Knuth recalls,  this ``was not only a famous unsolved problem for many years, it also was worked on by all of the leading researchers in the field, worldwide."   


\section{Bilinear Computational Problems and  Bilinear Algorithms}\label{sblnalg} 

 
 A natural  framework for their effort was the class of  noncommutative bilinear algorithms, also called just {\em bilinear algorithms}.
Such an algorithm solves a {\em bilinear computational problem} of the evaluation of 
a set of bilinear forms $c_h$, $h=1,\dots,\gamma$, in two sets  of variables.
  MM$(k,m,n)$ is a special case of this  problem where one evaluates
 a set of $kn$ bilinear forms $c_{i,h}$
such that $C=(c_{i,h})_{i,h=1}^{k,n}=AB$
for a pair of input matrices 
$A=(a_{i,j})_{i,j=1}^{k,m}$
and $B=(b_{j,h})_{j,h=1}^{m,n}$,
but it is more convenient to 
study the general case first.

Let these sets fill two vectors
${\bf a}=(a_i)_{i=1}^{\alpha}$
and ${\bf b}=(b_j)_{j=1}^{\beta}$,
let $T=(t_{i,j,h})_{i,j,h=1}^{\alpha,\beta,\gamma}$ denote the  {\em 3-dimensional tensor}  filled
with constants $t_{i,j,h}$
from  a fixed ring
(e.g., integers, rational, real, or complex numbers, or matrices filled with such numbers), and represent that problem as follows,
\begin{equation}\label{eqbln}    
 c_{h}({\bf a},{\bf b})=\sum_{i,j=1}^{\alpha,\beta}t_{i,j,h}a_ib_j,~{\rm
for}~h=1,\dots,\gamma.
\end{equation}  

A {\bf  bilinear algorithm} $\mathbb {BA}$ solves this {\em problem of size}
$(\alpha,\beta,\gamma)$ 
by successively computing 

(i) two sets of $2r$ linear forms,
$l_q=l_q({\bf a})$ and $l_q'=l_q'({\bf b})$, $q=1,\dots,r$, in the variables
$a_1,\dots$,$a_{\alpha}$ and 
$b_1,\dots$,$b_{\beta}$, which are
 the coordinates 
 of the vectors 
${\bf a}=(a_i)_{i=1}^{\alpha}$ and 
${\bf b}=(b_j)_{j=1}^{\beta}$, respectively, 

(ii) $r$ pairwise products $l_1l_1,\dots,l_rl_r'$, and

(iii) the $\gamma$ bilinear forms $c_{1}({\bf a},{\bf b}),\dots,c_{\gamma}({\bf a},{\bf b})$ as 
$\gamma$ linear combinations of these products.

\medskip

The straightforward VV and MV
and the algorithms of Examples \ref{ex0} and  \ref{ex1} for $2\times 2$ MM are examples of bilinear algorithms for MM.
At the end of this section we recall 
two other bilinear algorithms for two 
important and popular bilinear problems 
(see Examples \ref{excmplpr} and \ref{excmplpr}). 
 
\medskip

The number $r$ of bilinear multiplications 
at stage (ii) is called the {\em rank} of the algorithm.

The constant coefficients in parts (i), (ii), and (iii) form three matrices
$U=(u_i^{(q)})_{i,q=1}^{\alpha,r}$,
$V=(v_j^{(q)})_{j,q=1}^{\beta,r}$, 
and 
 $W=(w_h^{(q)})_{h,q=1}^{\gamma,r}$,
such that
\begin{equation}\label{eqblnr}
c_h=\sum_{q=1}^{r}w_h^{(q)}l_ql'_q,~
{\rm for}~
l_q=\sum_{i=1}^{\alpha}u_i^{(q)}a_i,~
 l'_q=\sum_{j=1}^{\beta}v_j^{(q)}b_j,~h=1,\dots,\gamma,~
{\rm and}~q=1,\dots,r.
\end{equation}
 

 The {\em rank of a bilinear computational problem} 
is the minimal rank of bilinear algorithms for that problem.
It depends on the field of constants,
e.g., can be different 
for real and complex constants. 

\medskip

A bilinear algorithm  $\mathbb {BA}$ above
 performs 
 
(i) $r$ bilinear  multiplications
of $l_q$ by  $l'_q$ for
$q=1,\dots,r$; 

(ii) $(\alpha+\beta+\gamma)r$
 multiplications by scalars $u_i^{(q)}$, $v_j^{(q)}$,
and $w_h^{(q)}$, for $i=1,\dots,\alpha$; $j=1,\dots,\beta$;
$h=1,\dots,\gamma$, and $q=1,\dots,r$; 

(iii)  $(\alpha-1)r$ 
additions of scaled variables $a_i$,  $i=1,\dots,\alpha$;
$(\beta-1)r$ additions
of scaled variables $b_j$, $j=1,\dots,\beta$,
and $(r-1)\gamma$ additions 
of scaled bilinear products $l_ql'_q$, $q=1,\dots,r$.

\medskip

These upper estimates decrease in the case of  sparse matrices $U$, $V$, and $W$.
Let $nnz(M)$ and $n_*(M)$ denote the numbers of  entries of a matrix $M$ that are nonzero and are neither of  0,1, and $-1$, respectively. Then the above bilinear algorithm performs  at most 
$n_*(U)+n_*(V)+n_*(W)$ scalar
multiplications and at most
$(nnz(U)-r)+(nnz(V)-r)+(nnz(W)-\gamma)$
scalar additions and subtractions. 

Fast bilinear algorithms for the following two bilinear problems  
enable fast practical complex and polynomial MM, respectively (see Section 
\ref{scmplpl}). 
 
\begin{example}\label{excmplpr}
{\rm Multiplication of two complex numbers.}
Evaluate 
the two bilinear forms
$a_1b_1-a_2b_2$  and $a_1b_2+a_2b_1$,
which represent the real and imaginary parts,
respectively, 
 of the product of two complex numbers $a_1+{\bf i}a_2$ and $b_1+{\bf i}b_2$.
The straightforward bilinear algorithm for this problem has rank 4, but here is a rank-3 algorithm: 

$l_1l_1'=a_1b_1$, $l_2l_2'=a_2b_2$, 
$l_3l'_3=(a_1+a_2)(b_1+b_2)$, 
 
$a_1b_1-a_2b_2=l_1l_1'-l_2l'_2$,
 $a_1b_2+a_2b_1=l_3l_3'-l_1l_1'-l_2l_2'$.

\end{example}

\begin{example}\label{explpr} {\rm Convolution.}
Compute  the coefficients of the {\em product} $c(x)=\sum_{h=0}^{m+n}c_hx^h$ 
 {\em of two polynomials} $a(x)=\sum_{i=0}^{m}a_ix^i$ and 
$b(x)=\sum_{j=0}^{n}b_jx^j$ or, equivalently, 
the 
{\em convolution} of the coefficient vectors of these
two polynomials, $c_h=\sum_{g=0}^ka_gb_{h-g}$,
for $h=0,\dots,m+n$, where $a_i=b_j=0$ for $i>m$ and $j>n$.
The straightforward  algorithm solves this problem by applying $(m+1)(n+1)$ scalar multiplications and $(m+1)(n+1)-m-n-1$
scalar additions, but FFT-based bilinear algorithm uses just 
$O((m+n)\log(m+n))$ arithmetic operations
(see Appendix \ref{scnvfft}). 
\end{example}

We refer the reader to
\cite{F72}, \cite{F72a},
\cite{P72}, \cite{BD73}, \cite{HM73}, \cite{S73}, \cite{P74}, \cite{BD76}, and \cite{BD78} 
for the early study of bilinear algorithms and to \cite[Section 2.5]{BM75} for its
 concise exposition, and to \cite[part 3 of Theorem 1]{P72}, \cite{R02},
and \cite[part 3 of Theorem 0.1]{P14}
for some results on the reduction from non-bilinear MM to bilinear MM.


\section{Tensor Representation of Bilinear  Algorithms and Tensor Product}\label{stnsrprpr} 


Observe that
 a bilinear algorithm $\mathbb {BA}$ of rank $r$  of Section \ref{sblnalg}
can be equivalently represented as a rank-$r$ decomposition of the tensor $T=(t_{i,j,h})_{i,j,h}$:
\begin{equation}\label{eqtnsrdcp}
t_{i,j,h}=\sum_{q=1}^ru_i^{(q)}v_j^{(q)}w_h^{(q)}~{\rm for}~
i=1,\dots,k;~j=1,\dots,m;~h=1,\dots,n.
\end{equation}
This implies that {\em the rank of a bilinear computational problem is precisely the rank of its tensor.}

Now suppose that two tensors $T=(t_{i,j,h})_{i,j,h=1}^{\alpha,\beta,\gamma}$
and $T'=(t'_{'i,j',h'})_{i',j',h'=1}^{\alpha',\beta',\gamma'}$ define two sets of bilinear forms
of the sizes $(\alpha,\beta,\gamma)$ and 
 $(\alpha',\beta',\gamma')$, respectively,
 and define another set of bilinear forms of the size 
  $(\alpha\alpha',\beta\beta',\gamma\gamma')$ by    
   the {\em tensor product}
\begin{equation}\label{eqtnsrprd}
T\otimes T'=
(t_{i,i',j,j',h,h'})_{i,i',j,j',h,h'=1}^{k,m,m,n,n,k}.
\end{equation}

\begin{theorem}\label{thtnsrprd}
Given two tensors 
$T=(t_{i,j,h})_{i,j,h}$ 
 of rank $r$ and $T'=
(t'_{i',j',h'})_
{i',j',h'}$,
 of rank $r'$, such that 
$t_{i,j,h}= 
\sum_{q=1}^{r}u^{(q)}_iv^{(q)}_jw^{(q)}_h$
for all $i$, $j$, and $k$ and 
$t'_{i'^{(q)},j'^{(q)},h'^{(q)}}=\sum_{q'=1}^{r'}u'^{(q)}_{i'}v'^{(q)}_{j'}w'^{(q)}_{h'}$  for all $i'$, $j'$, and $k'$,
the tensor product 
$T\otimes T'=(t_{i,i',j,j',h,h'})_{i,i',j,j',h,h'}$ has 
rank at most $rr'$.
\end{theorem}
\begin{proof}
Decompose the tensor $T\otimes T'$ by using the equations
$$t_{i,i',j,j',h,h'}=\Big (\sum_{q=1}^{r}u^{(q)}_iv^{(q)}_jw^{(q)}_h\Big )~
\Big (\sum_{q'=1}^{r'}u'^{(q)}_{i'}v'^{(q)}_{j'}w'^{(q)}_{h'}\Big )=
\sum_{q,q'=1}^{r,r'}u^{(qq')}_{i,i'}v^{(qq')}_{j,j'}w^{(qq')}_{h,h'}$$ where $u^{(qq')}_{i,i'}=u_i^{(q)}u'^{(q')}_{i'}$, 
$v^{(qq')}_{j,j'}=v^{(q)}_jv'^{(q)}_{j'}$, $w^{(qq')}_{h,h'}=w^{(q)}_hw'^{(q)}_{h'}$ for
 all 6-tuples $(i,i',j,j',h,h')$.
\end{proof}


\section{Bilinear MM and  the Associated Tensors}\label{sblnalgmm} 


The tensor $T$ associated with the problem MM$(k,m,n)$ has  entries with subscripts 
 represented by three pairs
of integers
$(i,i')$, $(j,j')$, and $(h,h')$, rather than by three
 integers $i$, $j$, and $h$:   
\begin{equation}\label{eqtnsrdlt}
T=(t_{(i,i'),(j,j'),(h,h')})_{i,i',j',j',h,h'=1}^{k,m,m,n,n,k},~
t_{(i,i'),(j,j'),(h,h')}= 
\delta_{i',j}~\delta_{j',h}\delta_{h',i}~{\rm  for~all}~i,i',j,j',h,~{\rm and}~h'.,
\end{equation}
Here and hereafter
\begin{equation}\label{eqdlt}
\delta_{q,q}=1,~\delta_{q,s}=0~{\rm if}~
q\neq s.
\end{equation}

We can represent a bilinear algorithm  of rank $r$
for  the computation of the matrix product $C=AB$  
 by the following equations: 
\begin{equation}\label{eqblnmm}
c_{i,h}=\sum_j a_{i,j} b_{j,h}=\sum_{q=1}^rw_{h,i}^{(q)}l_ql'_q~{\rm for}~i=1,\dots,k;~h=1,\dots,n.
\end{equation}
Here $l_q$ and $l_q'$ are linear forms 
in the entries of the matrices $A$ and $B$
(cf. (\ref{eqblnr})),
\begin{equation}\label{eqlnf}
l_q=l_q(A)=\sum_{i,j=1}^{k,m}u_{i,j}^{(q)}a_{i,j}~{\rm and}~ l'_q=l'_q(B)=\sum_{j,h=1}^{m,n}v_{j,h}^{(q)}b_{j,h},~
q=1,\dots,r,
\end{equation}
 and the algorithm is defined by a triple 
 of 3-dimensional tensors,
\begin{equation}\label{equvw}
U=\Big (u_{i,j}^{(q)}\Big )_{i,j,q=1}^{k,m,r},~
V=\Big (v_{j,h}^{(q)}\Big )_{j,h,q=1}^{m,n,r},~
W=\Big (w_{h,i}^{(q)}\Big )_{h,i,q=1}^{n,k,r}.
\end{equation}
We can rewrite the above expressions 
removing the links among the subscripts:
\begin{equation}\label{equvw'}
U=\Big (u_{i,i'}^{(q)}\Big )_{i,i',q=1}^{k,m,r},~
V=\Big (v_{j,j'}^{(q)}\Big )_{j,j',q=1}^{m,n,r},~
W=\Big (w_{h,h'}^{(q)}\Big )_{h,h',q=1}^{n,k,r}.
\end{equation}
Then simultaneous equations (\ref{eqblnmm}) and (\ref{eqlnf})
can be equivalently rewritten as follows
(cf. \cite{B70}),
\begin{equation}\label{eqbrnt}
\sum_{q=1}^r u_{i,i'}^{(q)}v_{j,j'}^{(q)}w_{h,h'}^{(q)}=
\delta_{i',j}~\delta_{j',h}~\delta_{h',i}~
{\rm for}~i,h'=1,\dots,k;i',j=1,\dots,m;j',h=1,\dots,n,
\end{equation} 

We can rewrite the tensor $T$
as 3-dimensional tensor by replacing every pair of subscripts by a single index, namely,
$(i,i')$ by $\bar i=i+mi'$ for $i=1,\dots,k$, $(j,j')$ by 
$\bar j=j+ni'$  for $j=1,\dots,m$,
and $(h,h')$ by $\bar h=h+kh'$  for $h=1,\dots,n$,
so that 
\begin{equation}\label{eqtnsrbr}
t_{(i,i'),(j,j'),(h,h')}=t_{\bar i,\bar j,\bar h}~{\rm for}~\bar i=1,\dots,km;\bar j=1,\dots,mn;\bar h=1,\dots,nk.
\end{equation}

 We can similarly write 
\begin{equation}\label{eqbrntbr}
\sum_{q=1}^r u_{\bar i}^{(q)}v_{\bar j}^{(q)}w_{\bar h}^{(q)}=
t_{\bar i,\bar j,\bar h}~
{\rm for}~\bar i=1,\dots,km;\bar j=1,\dots,mn;\bar h=1,\dots,nk.
\end{equation}

 
\section{Recursive Bilinear Algorithms for MM. Exponents of MM}\label{srbmm} 


The {\em tensor product construction}
of equation (\ref{eqtnsrprd})
provides useful insight into recursive algorithm for MM.
Given the problem MM$(k,m,n)$ of the computation of matrix product $C=AB$,
we can fix a triple of positive integers
$(k',m',n')$ substitute matrices of sizes $k'\times m'$,  $m'\times n'$, and 
 $k'\times n'$ for the entries of the matrices
 $A$, $B$,  and $C$, respectively, and arrive at the problem MM$(kk',mm',nn')$.
Equivalently we can define this problem
by its tensor, which is the product $T\otimes T'$ of 
the two tensors $T$ and $T'$ associated with 
the two problems  MM$(k,m,n)$
and  MM$(k',m',n')$, respectively. 
Apply
Theorem \ref{thtnsrprd} 
and obtain  
\begin{equation}\label{eqrnkprmm}
{\rm rank(MM}(kk',mm',nn'))\le
{\rm rank(MM}(k,m,n))\cdot{\rm rank(MM(}k',m',n')).
\end{equation}
By recursively applying  inequality (\ref{eqrnkprmm})
for $k=k'=m=m'=n=n'=2^i$, for $i=1,2,\dots$
we can bound the ranks in  
  recursive extensions of the algorithms of Examples \ref{ex0} and  \ref{ex1} for $2\times 2$ MM.

Next we  generalize the recursive processes based on Examples \ref{ex0} and \ref{ex1} --
we define {\em recursive bilinear  algorithms}
based on any  bilinear algorithm for MM$(n)$
of a fixed  $n$.  
\begin{theorem}\label{threcmm}
Given 
a bilinear 
algorithm of rank $r$ for  $n\times n$ MM, 
 one can perform $K\times K$ MM for all $K$ by
 using at most 
$c\cdot K^{\omega_{n,r}}$ scalar arithmetic operations
for a fixed $c$ independent of $K$
and for the exponent  $\omega_{n,r}=\log_n(r)$.
\end{theorem}
\begin{proof}
Substitute $n\times n$ matrices for variables,  re-apply the algorithm recursively, and in $d$  steps,  for any $d$, arrive at a  bilinear 
algorithm of rank $r^d$ for 
 $n^{d}\times n^{d}$ MM. Then recall that a linear operation of multiplication of a $q\times q$ matrix by a scalar as well as an addition or subtraction of a pair of $q\times q$ matrices can be performed in $q^2$ scalar arithmetic operations and deduce that  the arithmetic cost of performing all linear operations involved in the algorithm stays within the claimed bound.
 \end{proof}
By minimizing $\omega_{n,r}=\log_n(r)$ over the
 ranks $r$ of 
all bilinear algorithms for  
 $n\times n$ MM define
\begin{equation}\label{eqexpnr0}
\omega_{n}=\min_r~~\omega_{n,r}.
\end{equation}
Then, by minimizing $\omega_{n,r}$ over all 
integers $n$ not exceeding a fixed integer $K$,
 define the exponent
 \begin{equation}\label{eqexpfsbl}
\omega_{\le K}=\min_{n\le K} \omega_{n}.
\end{equation} 
For $K=\infty$ obtain
 the {\em universal or theoretical exponent of MM}
\begin{equation}\label{eqexpninf}
\omega=\inf_{n\le \infty} \omega_{n}. 
\end{equation}

Here and hereafter (except for Section \ref{simplsymbint})  we consider MM over the fields of real and complex numbers, but
the presented algorithms can be defined 
over other fields and rings as well
and in some cases (to a more limited extent) over semirings (see 
\cite{DP09}, \cite{Y09}, \cite{LG12}, 
and the bibliography therein). Over the fields 
the theoretical exponent $\omega$ only depends on the field characteristic  
\cite[Theorem 2.8]{S81}, while
the hidden overhead constants can vary greatly even 
over the fields having the same characteristic.



\section{To the exponent 2.78 by means of trilinear aggregation}\label{stragexp}

  
Breaking Strassen's barrier of $\log_2(7)\approx 2.8074$ for
$\omega$  was considered to be almost in hands in 1969, but this goal of
``literally all the leading researchers in the field, worldwide" 
 has remained a dream for almost a decade.
 
If one could build a recursive process on
the algorithm of Example \ref{exw0},
 then the  dream would have come true even well before Strassen's discovery of \cite{S69}. Indeed the  algorithm of this example has rank $r=n^3/2 +n^2$
for any even $n$, e.g., has rank $r=144$ for $n=6$.
Substitute these data into the equation $\omega=\log_n(r)$ 
and obtain 
$$\omega\le \log_6(144)\approx 2.7737<\log_2(7)\approx 2.8074.$$ 
 Theorem \ref{threcmm}, however,
cannot be applied here because MM is not commutative, and so the
substitution of matrices for the variables $a_i$ and  
$b_j$ would have  violated the basic identities of Example \ref{exw0}.
For example, we cannot apply the equation $v_{22}a_{22}=u_{22}v_{22}$  
if  $u_{22}$ and $v_{22}$ are matrices, 
e.g., if $u_{22}=\begin{pmatrix}1&0\\0&0\end{pmatrix}$ and 
$v_{22}=\begin{pmatrix}0&1\\0&0\end{pmatrix}$
because
$$
\begin{pmatrix}1&0\\0&0\end{pmatrix}~
\begin{pmatrix}0&1\\0&0\end{pmatrix}=\begin{pmatrix}0&1\\0&0\end{pmatrix}\neq
\begin{pmatrix}0&1\\0&0\end{pmatrix}~\begin{pmatrix}1&0\\0&0\end{pmatrix}=\begin{pmatrix}0&0\\0&0\end{pmatrix}.$$
The algorithm  of Example \ref{exw0}, based on (\ref{eqcmmm}), belongs to the class 
of {\em commutative  bilinear}  or {\em quadratic} algorithms. They only differ from 
 non-commutative bilinear algorithms at the stage of nonlinear multiplications $l_ql_q'$, $q=1,\dots,r$; they multiply 
pairs of linear forms 
 $l_q=l_q({\bf a},{\bf b})$ and 
 $l_q'=l_q'({\bf a},{\bf b})$ in the coordinates of both input vectors ${\bf a}$ and ${\bf b}$, but this difference turned out to be crucial when we try to apply the algorithm recursively. 

 Strassen's record of \cite{S69} would have fallen if one performed  $2 \times 2$ MM by using six bilinear multiplications,
but  
 \cite{HK69}, \cite{HK71}, \cite{BD78}
proved  that 7 is the sharp lower bound on the rank of 
 $2\times 2$ MM.  

Actually 15 is the sharp lower bound
on the number of additions and subtractions in all bilinear algorithms 
of rank 7 for  $2\times 2$ MM (cf. \cite{P76}, \cite{B95}), and moreover the following theorem defines explicit expressions
for all bilinear algorithms of rank 7 for  $2\times 2$ MM
appeared in \cite[Theorem 3]{P72} (see also \cite{dG78} 
and \cite[Theorem 0.3]{P14}). 

Two bilinear algorithms, both of rank $r$,  for the same problem of MM$(k,m,n)$
(see (\ref{equvw})),  defined by 
two triples $\{U,V,W\}$ and $\{\overline {U},\overline{V},\overline{W} \}$, respectively,
are said to be {\em equivalent} 
to one another if 
$$\overline{u}_{i,j}^{(q)}=
\Sigma_{v,\kappa}\sigma_{i,v} \nabla_{j,\kappa} u_{v,\kappa}^{(t(q))},~\overline{v}_{g,h}^{(q)}=
\Sigma_{v,\kappa}\lambda_{v,g} \mu_{h,\kappa} v_{v,\kappa}^{(t(q))},~{\rm and}~\overline{w}_{l,q}^{(q)}=
\Sigma_{v,\kappa}\gamma_{v,l}\beta_{\kappa, q}w_{v,\kappa}^{(t(q))},$$
where the matrices in the three pairs $(\sigma_{i,v})$ and $(\gamma_{v,l} ),(\nabla_{j,\kappa})$ and $(\lambda_{v,g} )$,
and $(\mu_{h,\kappa})$ 
and $(\beta_{\kappa, q} )$ are the inverses of one another; $1\leq t(s)  \leq r$; 
$t_{q_{1}} \neq t_{q_{2}}$ if $q_{1} \neq {q_{2}}$, and all $t(q)$ are integers. 
\begin{theorem}\label{th3} 
Every bilinear algorithm of rank 7 for  $2\times 2$ MM is
equivalent to the algorithms of Examples \ref{ex0} and \ref{ex1}.
\end{theorem}

We refer the reader to 
the paper 
 \cite{MR14} for
the current record lower bounds on the 
rank of $MM(n)$ for all $n$, to  the papers
 \cite{HK69}, \cite{HK71},  \cite[Theorem 1]{P72}, 
\cite{BD73},   \cite{BD78},   
 \cite{B89},  \cite{B99},  \cite{B00},
 \cite{RS03}, \cite{S03}, and \cite{L14}
for some earlier work in  this direction, and to the papers
\cite{DIS11} and \cite{S13} for various lower and upper bounds
on the arithmetic complexity and the ranks of rectangular MM of smaller sizes.
 
Since the study of MM(2) could not help  decrease the exponent $\log_2(7)$, 
the researchers tried to devise a bilinear algorithm of rank 21 for  MM(3) because 
$\log_3(21)<\log_2(7)\approx 2.8074$.
This turned out to be hard, and we still cannot
perform  $3 \times 3$ MM by using less than 23 bilinear multiplications. 
 
The exponent $\log_2(7)\approx 2.8074$
was decreased only in 1978, when
the paper \cite{P78} presented 
a bilinear algorithm of rank 143,640 for $70\times 70$ MM. This  implied 
the exponent $\omega=\log_{70} (143,640)<2.7962$
for $MM(n)$ and consequently for $n\times n$ matrix inversion, Boolean $MM(n)$, 
and various other well-known computational problems  linked to MM and partly listed in Section \ref{sblnpex}.

The paper \cite{P78} has
extended some novel techniques of 
the paper \cite{P72} of 1972,  
published in Russian\footnote{Until 1976 the  author of \cite{P66} and \cite{P72} lived in the Soviet Union. 
From 1964 to 1976 he has been working
 in Economics in order
to make his living
and has written the papers \cite{P66} and \cite{P72} in his spare time.}
and translated into English only in 2014
in \cite{P14}. Namely the paper  \cite{P72}
has accelerated the straightforward MM by combining
 {\em trilinear interpretation} of bilinear algorithms
and the {\em aggregation} method. By following \cite{P78}
we call this combination
{\em trilinear aggregation}
and  briefly cover it in the next two sections.
By refining 
 trilinear aggregation of
 \cite{P72}
the papers \cite{P78},  \cite{P79},  \cite{P80a},  \cite{P81}, and
 \cite{P82} proposed various algorithms that
 further accelerated MM. In particular
 the paper  \cite{P82} yielded   
  the exponent 
\begin{equation}\label{eqexprlst}
\omega_{44}\le 2.7734,
\end{equation} 
  and this
  still remains the record  exponent $\omega_{\le K}$ for feasible MM. 
  

The technique of trilinear aggregation  has 
been recognized for its impact
on the decrease of the MM exponent,
but the paper \cite{P72} was also a 
landmark in the study of multilinear and tensor decompositions.
Such decompositions  introduced by Hitchcock
in 1927  received little attention except for
a minor response in  1963--70 with half of a dozen
papers  in  
 the psychometrics literature. 
The paper \cite{P72} of 1972 
 provided the earliest known application of nontrivial 
multilinear and tensor decomposition to fundamental matrix computations,
now a popular flourishing area in
linear and multilinear algebra with a wide range of 
important applications to modern computing 
(see  \cite{T03}, \cite{KB09},  \cite{OT10}, 
\cite{GL13},
and the bibliography therein).


\section{Trilinear Representation and 
Dual Bilinear 
Algorithms}\label{strld} 

 
Trilinear representation of  a bilinear algorithm enables  transparent demonstration of the technique of trilinear aggregation. Otherwise it  is equivalent and quite similar to its tensor representation.

Let a bilinear algorithm
of rank $r$ be represented by 
equations  (\ref{eqblnr}).
Multiply them by new variables $d_h$,
 sum the products in $h$, and arrive at 
 the following
  representation of the algorithm
  as a decomposition of a trilinear 
  form,
\begin{equation}\label{eqtrlntns}
\sum_{i,j,h=1}^{k,m,n} t_{i,j,h}a_{i} b_{j}d_{h}=
\sum_{s=1}^r l_q({\bf a})l'_q({\bf b})l''_q({\bf d})
\end{equation} 
for
 $l_q=l_q({\bf a})$ and $l'_q=l'_q({\bf b})$ of (\ref{eqblnr}), $l_q''=l''_q({\bf d})=\sum_{h=1}^{n}w_{h}^{(q)}d_{h}$, and  $q=1,\dots,r$.
By equating the coefficients of the variables
$d_h$ on both sides of this trilinear decomposition  we come back to the original 
bilinear representation  (\ref{eqblnr}) of the same algorithm, and we can obtain its two alternative dual bilinear representations by 
equating the coefficients of the variables 
$a_{i}$ and $b_j$ instead.

 Here is a simple example of the trilinear representation of the bilinear algorithm of Example \ref{excmplpr}.

\begin{example}\label{extrcopr} 
{\em A trilinear
decomposition 
of rank 3 for 
 multiplication
of two complex numbers.} 
$$a_1b_1d_1-a_2b_2d_1+a_1b_2d_2+a_2b_1d_2=
a_1b_1(d_1-d_2)-a_2b_2(d_1+d_2)+(a_1+a_2)(b_1+b_2)d_2.$$
\end{example}
  
  By equating the coefficients of the variables
$d_h$ on both sides we come back to the bilinear algorithm of Example \ref{excmplpr}.
 By equating the coefficients of  
 $a_1$ and $a_2$  on both sides of this equations
or alternatively the coefficients of 
 $b_1$ and $b_2$
  on their both sides, we arrive at two alternative bilinear algorithms of rank 3 for computing the product of two complex numbers. They are close to one another but not to the algorithm of Example \ref{excmplpr}. We display just one of the two.
  
  \begin{example}\label{exblncopr2} 
{\em A distinct
trilinear
decomposition 
of rank 3 for 
 multiplying
two complex numbers.} 

$l_1=b_1$, $l_1'=d_1-d_2$,
 
$l_2=b_2$, $l_2'=d_1+d_2$,

$l_3=b_1+b_2$, $l_3'=d_2$,

 $b_1d_1+b_2d_2=l_1l_1'+l_3l_3'$,

$b_1d_2-b_2d_1=l_3l_3'-l_2l_2'$.
\end{example}

 The book \cite{W80} demonstrates the power of 
duality technique in devising some
efficient bilinear algorithms
for FIR-filters and multiplication of complex numbers and polynomials.

In the rest of this section we apply and extend the above discussion to 
the special case of the  algorithms   for the problem
 MM$(k,m,n)$ of multiplying two matrices 
 $A=(a_{i,j})_{i,j=1}^{k,m}$ and $B=(b_{j,h})_{j,h=1}^{m,n}$. We 
  can represent such an algorithm by means of the following trilinear decomposition,
\begin{equation}\label{eqtrc}
{\rm Trace}(ABD)=\sum_{i,j,h} a_{i,j} b_{j,h}d_{h,i}=
\sum_{s=1}^r l_q(A)l'_q(B)l''_q(D)~{\rm for}~
l_q''=l''_q(A)=\sum_{h,i=1}^{n,m}w_{h,i}^{(q)}d_{h,i},
\end{equation} 
 $l_q$ and $l_q'$ of (\ref{eqlnf}), and 
 $q=1,\dots,r$.
Here $D=(d_{h,i})_{h,i=1}^{n,k}$
is an auxiliary  $n\times k$ matrix and 
 Trace$(M)=\sum_i m_{i,i}$ denotes the trace of a matrix $M=(m_{i,j})_{i,j}$.

\begin{example}\label{extrmm2}
{\em A trilinear vesion of Strassen's bilinear algorithm of Example \ref{ex0} 
 for $MM(2)$.}  \\
$~~~~~~~~~ \sum_{i,j,h=1}^2a_{i,j}b_{j,h}d_{h,i}=\sum_{s=1}^7l_sl_s'l_s''$,~
$l_1l_1'l_1''=(a_{11}+a_{22})(b_{11}+b_{22})(d_{11}+d_{22})$, \\
$l_2l_2'l_2''=(a_{21}+a_{22})b_{11}(d_{21}-d_{22})$,
$l_3l_3'l_3''=a_{11}(b_{12}-b_{22})(d_{12}+d_{22})$,
$l_4l_4'l_4''=(a_{21}-a_{11})(b_{11}+b_{12})d_{22}$,\\
$l_5l_5'l_5''=(a_{11}+a_{12})b_{22}(d_{12}-d_{11})$,
$l_6l_6'l_6''=a_{22}(b_{21}-b_{11})(d_{11}+d_{21})$,
$l_7l_7'l_7''=(a_{12}-a_{22})(b_{21}+b_{22})d_{11}$.
\end{example}

We can come back to 
the original bilinear algorithm of Example
\ref{ex0}
for  the $2\times 2$ matrix product $AB$ by equating 
the coefficients of the variables $d_{h,i}$ on both sides of a trilinear
decomposition. By equating the coefficients of 
the variables $a_{i,j}$ also on both sides or alternatively of $b_{j,h}$ on both sides,
we can obtain two dual bilinear algorithms.
In this case the three dual algorithms differ little from each other,
but let us be given a bilinear algorithm of a rank $r$ for rectangular MM$(k,m,n)$. 
Then we arrive at the dual algorithms of rank $r$ for the problems  $MM(n,m,k)$, and $MM(k,n,m)$ as well
(cf.\ \cite[part 5 of Theorem 1]{P72},  \cite{BD73},
\cite{HM73}, \cite{P74}).
A bilinear algorithm  for 
MM$(k,m,n)$ can be readily extended to three other dual algorithms of the same rank 
for the  three other problems  MM$(m,k,n)$,
 $MM(n,m,k)$, and
 $MM(k,n,m)$ because we can interchange the subscripts of the variables. 


 The number of  linear  operations
  (unlike the rank) can differ in the transition among the three dual algorithms, and this can possibly be exploited  for minimizing this number.
  

Here is a useful combination of duality with tensor product construction.
Given a trilinear decomposition  of rank $r$ 
for  the problem of MM$(k,m,n)$, obtain that the dual problems of MM$(m,n,k)$ and 
 MM$(n,k,m)$ also have rank $r$.
Apply bound (\ref{eqrnkprmm}) and obtain that
the problem MM$(kmn,kmn,kmn)$ has rank at most $r^3$.
Now apply Theorem \ref{threcmm} for $n$ replaced by $kmn$ and obtain the following  result, which first appeared as
 claim 1 of \cite[Theorem 1]{P72}.
 \begin{theorem}\label{thdual} 
Given a bilinear or trilinear algorithm 
of rank $r$ for   $MM(k,m,n)$
and any 4-tuple of integers $k$, $m$, $n$, and $r$ such that $kmn>1$,
one can perform $MM(K)$ by using 
$cK^{\omega}$
arithmetic operations for any $K$, for
 $\omega=\omega_{k,m,n,r}=3\log_{kmn}(r)$,
and for a constant $c$ independent of $K$.
\end{theorem}


\section{Trilinear Aggregation and Disjoint MM}\label{stragg}


Aggregation  technique is well-known in 
business, economics, computer science, telecommunication,
natural sciences, medicine, and statistics. 
The idea is to  mass together or cluster numerous independent but similar units
into much fewer aggregates. Then the study is simplified
but is supposed to
characterize all the  units
either directly or via 
disaggregation techniques.
Such aggregation/disaggregation processes
 in \cite{MP80} served as a springboard 
 for the emergence of 
the field of {\em Algebraic Multigrid},  now quite popular. 
 
Aggregation/disaggregation techniques are behind
the acceleration of MM in Example~\ref{exw0},
 preceded by similar application of this
technique to polynomial evaluation with preprocessing of 
coefficients \cite{P66}, \cite{K97}. 
The papers \cite{P72} and  \cite{P78} 
apply
aggregation in order to compress the decomposition
of the trilinear form ${\rm Trace}(ABC)$
by playing with the shared  subscripts of distinct variables.
Other implementations of this technique appeared
 in \cite{P79},   \cite{P80},  
 \cite{P81}, \cite{P82}, and \cite{LPS92}.
 
 For demonstration of these techniques, consider
{\em disjoint MM} of computing two independent matrix products $AB$ and $UV$ 
represented by the trilinear form
$${\rm Trace}(ABD+UVW)=\sum_{i,j,h=1}^{k,m,n}(a_{i,j}b_{j,h}d_{h,i}+u_{j,h}v_{h,i}w_{i,j}).$$
For $k=m=n$, we would seek a pair of disjoint $n\times n$ matrix products, which is quite a realistic task  
in computational practice.

For each triple $i,j,h$ define the aggregate $(a_{i,j}+u_{j,h})(b_{j,h}+v_{h,i})(d_{h,i}+w_{i,j})$ of two monomials $a_{i,j}b_{j,h}d_{h,i}$  and $u_{j,h}v_{h,i}w_{i,j}$ and let
$$T=\sum_{i,j,h=1}^{k,m,n}(a_{i,j}+u_{j,h})(b_{j,h}+v_{h,i})(d_{h,i}+w_{i,j})$$
denote the sum of the $kmn$
aggregates.
 Let 
 $$T_1=\sum_{i,j=1}^{k,m}a_{i,j}s_{i,j}w_{i,j},~ 
T_2=\sum_{j,h=1}^{m,n}u_{j,h}b_{j,h}r_{j,h},~ {\rm and}~T_3=\sum_{h,i=1}^{n,k}q_{i,h}v_{h,i}d_{h,i}$$
denote three groups of  correction terms
where $$q_{i,h}=\sum_{j=1}^{m}(a_{i,j}+u_{j,h}),~
s_{i,j}=\sum_{h=1}^{n}(b_{j,h}+v_{h,i}), 
{\rm and}~r_{j,h}=\sum_{i=1}^{k}(d_{h,i}+w_{i,j}).$$
Then the equation 
\begin{equation}\label{eq2aggr}
{\rm Trace}(ABD+UVW)=T-T_1-T_2-T_3
\end{equation} defines
a trilinear decomposition
 of rank $mnp+mn+np+pm$ (versus the rank $2mnp$  of the  straightforward algorithm).
Table~\ref{tabaggr2} displays this decomposition in compressed form.

\begin{table}[ht] 
\caption{Aggregation/disaggregation of a pair of MM terms.}
\label{tabaggr2}
  \begin{center}
\begin{tabular}{| c | c |c |}
      \hline
 $a_{i,j}$ & $y_{j,h}$ & $d_{h,i}$  \\ \hline
 $u_{j,h}$ &  $v_{h,i}$   & $w_{i,j}$ \\ \hline
 \end{tabular}
\end{center}
\end{table}

Sum the two entries in each column of the  table, multiply the three products together, and obtain an aggregate. 
Multiply  together
the three entries in each row of the  table and obtain the two output terms 
$a_{i,j}b_{j,h}d_{h,i}$
and $u_{j,h}v_{h,i}w_{i,j}$.
The cross-products of other triples 
of the table define six correction terms.
Their sum over all $n^3$ triples of indices 
$i,j$, and $h$ has rank $2(km+mn+nk)$.
By  
subtracting this sum from the 
sum of all $kmn$ aggregates,
we decompose $2kmn$ terms of ${\rm Trace}(ABD+UVW)$
 into the sum of $kmn+2(km+mn+nk)$
terms. 
For $m=n=p=34$ this implies a decomposition of 
rank $n^3+6n^2$ for a pair of disjoint $MM(n)$, versus the rank $2n^3$ of the straightforward decomposition.
 
Demonstration of 
the power  of trilinear aggregation
can be made more transparent
 for disjoint MM, whose natural link to 
trilinear aggregation 
has  been shown in \cite{P84},  \cite[Section 5]{P84a},
 \cite[Section 12]{P84b}, and \cite{LPS92}.
The known constructions 
 for pairs  of disjoint 
 $n\times n$ MM, however,
have been  extended to a single $MM(n)$ for even $n$. 
In particular
the paper \cite{P72} presented
a trilinear decomposition of rank $0.5 n^3+3n^2$ for $MM(n)$ and any even $n$
similar to 
the above decomposition of ${\rm Trace}(ABD+UVW)$.
This implied the $MM$ exponent $\log_n (0.5n^3+3n^2)$,
which is less than 2.85 for  $n=34$. 
 
The paper \cite{P78} presented
a trilinear decomposition
 of rank $(n^3-4n)/3+6n^2$ 
for  $MM(n)$, $n=2s$, and any positive integer $s$.
For $n=70$ this defines the MM exponent  2.7962.
Then again it is convenient to demonstrate 
this design for disjoint MM 
associated with a decomposition 
of the trilinear form ${\rm Trace}(XYZ+UVW+ABD)$.
The basic step is the  
aggregation/disaggregation defined by Table~\ref{tabaggr3}.

\begin{table}[ht] 
\caption{Aggregation/disaggregation of a triple of MM terms.}
\label{tabaggr3}
  \begin{center}
\begin{tabular}{| c | c |c |} \hline
 $x_{i,j}$ & $b_{j,h}$ & $z_{h,i}$  \\ \hline
 $u_{j,h}$ &  $v_{h,i}$   & $w_{i,j}$ \\ \hline
 $a_{h,i}$   & $b_{i,j}$ & $d_{j,h}$ \\ \hline
 \end{tabular}
\end{center}
\end{table}

Sum the $kmn$
aggregates 
$$(x_{i,j}+u_{j,h}+a_{h,i})(y_{j,h}+v_{h,i}+b_{i,j})(z_{h,i}+w_{i,j}+d_{j,h}),$$
subtract order of $n^2$ correction terms,  and
obtain a decomposition of rank $n^3+O(n^2)$
for $${\rm Trace}(XYZ+UVW+ABD),$$
versus the rank  $3n^3$ of the straightforward decomposition.
The trace represents 
three disjoint problems of $MM(n)$, that is,
the computation of
the three independent
$n\times n$ matrix products $XY$, $UV$,  and $AB$,
and we obtain a trilinear decomposition of rank $n^3+O(n^2)$
for this task.

With a little more work 
one obtains a similar 
trilinear decomposition of \cite{P78}
 of rank $(n^3-4n)/3+6n^2$, for any even $n$, and this implied 
the bound $\omega_{70}<2.7962$.
Refinements of this construction implied
smaller upper bounds (see Tables \ref{tab2} and \ref{tab2a}). In particular the algorithm of  \cite{P82} yielded the bound   
$\omega_{44}\le 2.7734$
of (\ref{eqexprlst}).


\section{Any Precision Approximation (APA)  Algorithms}\label{sapa}

  
Based on the following table we arrive at the technique of {\em Any Precision Approximation},\footnote{Hereafter we use the acronym {\em APA}.} which is quite efficient for fast symbolic MM  and for other symbolic algebraic computations.  
  
\begin{table}[ht] 
\caption{Aggregation/disaggregation of a pair of MM terms for Any Precision Approximation MM.}
\label{tabaggrapa}
  \begin{center}
\begin{tabular}{| c | c |c |}
      \hline
 $a_{i,j}$ & $b_{j,h}$ & $\lambda^2 d_{h,i}$  \\ \hline
 $\lambda u_{j,h}$ &  $\lambda v_{h,i}$   & $w_{i,j}$ \\ \hline
    \end{tabular}
\end{center}
\end{table}
For $\lambda=1$ this table  turns into 
 Table  \ref{tabaggr2} but for variable  
$\lambda$ helps us demonstrate the APA
technique    
 of  \cite{BCLR79},  \cite{BLR80}, and \cite{BCLR81}.
  

Let $\lambda\rightarrow 0$ and 
obtain  trilinear decomposition 
\begin{equation}\label{eqapa}
 {\rm Trace}(ABD+UVW)=\lambda^{-2}(S-T_1-T_2+O(\lambda)),
\end{equation}
where
$$S=\sum_{i,j,h=1}^{k,m,n}(a_{i,j}+\lambda u_{j,h})(b_{j,h}+
\lambda v_{h,i})(\lambda^2 d_{h,i}+w_{i,j}),$$ is the  sum 
of $kmn$ aggregates,

 $$T_1=\sum_{i,j=1}^{k,m}a_{i,j}q_{i,j}w_{i,j},~ {\rm and}~
T_2=\sum_{j,h=1}^{m,n}u_{j,h}b_{j,h}r_{j,h}$$
for
$$q_{i,j}=\sum_{h=1}^{n}(b_{j,h}+\lambda v_{h,i})~
{\rm and}~ r_{j,h}=\sum_{i=1}^{k}(\lambda^2 d_{h,i}+w_{i,j}).$$ 
The terms of order $\lambda$
vanish as $\lambda\rightarrow 0$.
Counting only the remaining monomials
on the right-hand side of (\ref{eqapa}),
we define the
 {\em  border rank} of the decomposition.
It is equal to $kmn+km+kn$, versus the larger rank
 $kmn+km+kn+mn$ of decomposition (\ref{eq2aggr}). 

Generally, given a  trilinear form $T$
(e.g., given by the 3-dimensional tensor of its coefficients), multiply it by 
$\lambda^d$ for a fixed nonnegative 
integer $d$ and define a
trilinear decomposition of the trilinear form 
$\lambda^d\cdot T$ with coefficients being polynomials in 
$\lambda$. Delete the terms of order $\lambda^{d+1}$
and higher  and call the number of the remaining terms the
 border rank of the decomposition and of the associated APA algorithm. 
 The minimal border rank  over all such APA algorithms for a fixed trilinear  form $\lambda^d\cdot T$ define its border rank.
All this can be readily restated for a set of bilinear forms replacing a single trilinear form $T$.

We can
equate the coefficients of the variables $d_{h,i}$
and $w_{i,j}$ in the trilinear APA decomposition above and arrive at the bilinear problem of the evaluation of 
two disjoint matrix products 
$AB$ and $UV$. The $kmn$  
trilinear aggregates turn into the $kmn$  
 bilinear  products
 $(a_{i,j}+\lambda u_{j,h})(b_{j,h}+
\lambda v_{h,i})$ for all $i$, $j$, and $h$.
Clearly, $$a_{i,j}+\lambda u_{j,h}\rightarrow a_{i,j}~ {\rm and}~b_{j,h}+
\lambda v_{h,i} \rightarrow b_{j,h}
~ {\rm as}~\lambda \rightarrow 0,$$  
but we must keep the terms   
$\lambda u_{j,h}$ and $\lambda v_{h,i}$
in the aggregates
in order to compute the matrix   $UV$.

This forces us to double the precision of the  
representation of
the multiplicands   
 $a_{i,j}+\lambda u_{j,h}$
and $b_{j,h}+
\lambda v_{h,i}$
compared to the precision of the representation
of the entries 
 $a_{i,j}$, $u_{j,h}$, $b_{j,h}$, and
$v_{h,i}$.

For example, suppose that $\lambda=2^{-s}$
for a sufficiently large integer $s$ and that
$a_{i,j}$, $u_{j,h}$,
 $b_{j,h}$, and
$v_{h,i}$ are $s$-bit integers in the range $[0,2^s)$.
Then $2s$ bits
are required in order to represent
each of the multiplicands   
 $a_{i,j}+\lambda u_{j,h}$
and $b_{j,h}+
\lambda v_{h,i}$.
If $s$ exceeds  the  
half-length of the computer word,
then using $2s$ bits in APA algorithms
would move us beyond the length of the computer word, 
e.g.,  beyond the IEEE standard double precision.

This can be  
costly in numerical computations, but
APA is  valuable in
symbolic computation where efficient tchniques such as Chinese remainder algorithm and $p$-adic lifting facilitate computations with long numbers (cf. our Section \ref{simplsymbint}).

The papers \cite{BCLR79},  \cite{BLR80}, and \cite{BCLR81} study border rank of MM and various other fundamental bilinear  computational
problems and show that  border rank is  quite frequently  smaller than their rank.


\section{Bini's  Construction}\label{sapaexp}


In the community of the Theory of Computing
APA algorithms 
 have been mostly and highly appreciated  as a tool for
decreasing the record upper bounds on
the theoretical exponent 
 $\omega$  of (\ref{eqexpninf}). 
Their power in this context is due to Bini's theorem in \cite{B80},
according to which APA decomposition of border
rank $r$ 
 and bilinear algorithms 
of rank $r$, both  for the same bilinear problem, 
define the same upper estimate 
for the exponent $\omega$. 
 Bini's argument demonstrates generally fruitful idea of operating with matrix 
polynomials and finally recovering some scalar matrices of their coefficients. Historically this became the springboard for the derivation of the MM  technique of  EXPAND, 
PRUNE,  and CONQUER.

Next we outline 
 Bini's argument. 

Multiply 
trilinear decomposition (\ref{eqapa}) of 
${\rm Trace}(ABD+UVW)$
 by 
the  variable $\lambda$
 and arrive at a decomposition of
$\lambda\cdot {\rm Trace}(ABD+UVW)$ 
 whose
coefficients are polynomials in $\lambda$ of degree at most $d=2$.
 Interpolate to ${\rm Trace}(ABD+UVW)$ from $d+1$ values of the polynomial.
The interpolation  increases the
 rank  by a factor of at most $(d+1)^2$,
that is, by  at most a factor of 9 
for $d=2$,\footnote{By using FFT we can interpolate to a polynomial from its values at $K$th roots of unity
by using $1.5K\log_2K+K$ arithmetic operations provided that $K=2^k>d$ for a positive integer $k$ \cite[Theorem 2.2.2]{P01}; for $d=k=2$, this implies  interpolation  factor 14.} 
and the resulting rank $9\theta\cdot (mkn+mk+kn)$ 
for a constant $\theta>1$ greatly exceeds the rank $2mkn$ of the straightforward algorithm
for ${\rm Trace}(ABD+UVW)$.  

We can overcome this deficiency, however, if
we recursively extend an APA algorithm to 
$MM(n^{2^q})$ for $q=0,1,\dots$. At the 
$q$th recursive step substitute $n^{2^q}\times n^{2^q}$ matrices for the input entries and observe that this squares both dimension $N$ of the $MM(N)$ and  border rank $br(N)$ but only doubles
the degree of the decomposition in $\lambda$. Hence in $q$ recursive steps 
$${\rm r}(N)\le {\rm INT}_q\cdot {\rm br}(N),~{\rm INT}_q=2^q(d+1)^2,~{\rm br}(N)=({\rm br}(n))^{2^q}$$
where ${\rm r}(N)$ denotes  rank,  
${\rm br}(N)$ and ${\rm br}(n)$ denote border ranks,
 $N=n^{2^q}$, and ${\rm INT}_q$
 is the  growth factor  of the border
 rank in its transition to the rank
 in $q$ recursive steps. 
Therefore
$$\omega_{N,{\rm br}(N)}\le \log_{N}({\rm r}(N))\le
\log_{N}({\rm INT}_q\cdot {\rm br}(N))=
\log_{n^{2^q}}({\rm INT}_q\cdot {\rm br}(n)^{2^q})=\log_n(({\rm INT}_q)^{1/2^q}{\rm br}(n)).$$
Observe that
$({\rm INT}_q)^{1/2^q}\rightarrow 1$ as 
$q\rightarrow \infty$, recall  bound (\ref{eqexpninf}), 
 and deduce Bini's estimate
$$\omega\le \lim_{q\rightarrow \infty}\omega_{N,{\rm br}(N)}\le
\lim_{q\rightarrow \infty}\log_{n}({\rm br}(n)).$$

The argument is readily extended to
APA of rectangular MM, that is,
  $\omega\le 3\log_{kmn}({\rm br}(k,m,n))$ 
where ${\rm br}(k,m,n)$ denotes the border rank of MM$(k,m,n)$
 for a fixed triple of
  $k$, $m$, and $n$.


\section{Sch{\"o}nhage's Construction. EXPAND, PRUNE,  and \\ CONQUER Algorithms}\label{sapaexpsch}


The above APA decomposition for 
${\rm Trace}(ABD+UVW)$ is associated with disjoint MM rather than MM,
but Sch{\"o}nhage in \cite{S81} has
extended Bini's theorem by proving that
the theoretical exponent $\omega$ of MM can be bounded  based on
the APA decomposition for  disjoint MM as follows.
\begin{theorem}\label{thsch} \cite{S81}.
The theoretical exponent $\omega$ of MM in (\ref{eqexpninf}) 
does not exceed  $3\tau$ if
there exists a bilinear or trilinear algorithm  of  rank $r$ or 
an 
APA algorithm of  border rank $r$
 for $s$ disjoint MM
problems of sizes $(k_i,m_i,n_i)$, for $i=1,\dots,s$,
such  that $\sum_{i=1}^s (k_im_in_i)^{\tau}=r$.
\end{theorem}

The theorem has interesting interpretation
in terms of the following {\em Direct Sum Conjecture} (first stated and then retracted by V. Strassen): 
{\em the rank of $s$ disjoint MM problems of sizes $(k_i,m_i,n_i)$, for $i=1,\dots,s$, is not less than $\sum_{i=1}^sr_i$ where $r_i$ denotes  the rank of the problem} 
MM$(k_i,m_i,n_i)$.

This conjecture is still open but becomes wrong if in its statement border ranks replace ranks. Indeed decomposition (\ref{eqapa}) for 
disjoint pair of MM$(k,m,n)$ and MM$(m,n,k)$ has border rank $knm+km +kn$, that is,
$mn+m+n$ for $k=1$. One can readily prove that
the border rank of each of MM$(1,m,n)$ and MM$(m,n,1)$ is $mn$ and then observe that
  $2mn>mn+m+n$ for $m=n>2$.\footnote{This  counter-example to the conjecture appeared in \cite{S81}; a little more involved one, based on  an APA-variant of the decomposition of Table \ref{tabaggr3}, appeared in \cite{P79}.}
  
  Nevertheless  Theorem \ref{thsch} can be equivalently stated as follows: {\em The  
  bound of  Theorem \ref{thsch} on the exponent $\omega$ of (\ref{eqexpninf})  
  cna be immediately verified if
 the border rank version of the Direct Sum Conjecture held true.} In \cite{S81}
 Sch{\"o}nhage proved the theorem 
 without assuming that the Conjecture is true. 
His proof is the 
simplest in the case where 
$k_i=m_i=n_i=n$, for $i=1,\dots,s$, and $s$ divides $r$, that is, where we are given an APA bilinear decomposition of  border rank $r=gs$ for $s$ disjoint problems  ${\rm MM}(n)$ and an integer $g$. Then substitute $n\times n$ 
matrices for the variables, reapply the algorithm for every $s$-tuple of bilinear multiplications,
obtain an APA bilinear algorithm of border rank
$r\cdot r/s=(r/s)^2s$ for 
$s$ disjoint problems of MM$(n^2)$,
and extend this process recursively.
$q$ recursive steps define an APA bilinear algorithm of border rank  $(r/s)^qs$ for  
$s$ disjoint problems of MM$(n^q)$. 
Prune the input keeping just a 
single problem MM$(n^q)$, apply  Bini's theorem,  and deduce from (\ref{eqexpninf}) that
$$ \omega\le \log_{(n^q)}\Big (s\cdot \Big (\frac{r}{s}\Big )^q\Big )=
\log_{n}\Big (\frac{r}{s}\Big )+
\frac{1}{q}\log_n(s).$$
For $q\rightarrow \infty$  obtain 
 Sch{\"o}nhage's bound 
$\omega\le \log_{n}(r/s)$ in this special case.

We can quite readily relax the assumption that $s$
divides $r$;    
furthermore by proceeding similarly to the proof of 
Theorem \ref{thdual}  
 we yield  extension to the case of
 $s\cdot{\rm MM}(k,m,n)$, where we are given an APA bilinear algorithm of a border rank $r$ for $s$ disjoint problems of rectangular ${\rm MM}(k,m,n)$.

Let us extend these results to the pair of disjoint problems of MM$(k,m,n)$ and 
MM$(m,n,k)$   represented in Tables \ref{tabaggr2} and \ref{tabaggrapa}. Let a basic APA algorithm have a border rank $r$.
Then $q$ recursive steps define an APA algorithm of border rank $r^q$
for $2^q$ disjoint MM problems. Grouping together the MM problems of the same size we 
obtain $q+1$ disjoint groups of MM problems 
$T_i\cdot(k^im^{q-i},m^in^{q-i},n^ik^{q-i})$
with binomial coefficients
$T_i=\begin{pmatrix}q\\i \end{pmatrix}$
for $i=0,1,\dots,q$.
Choose even $q$,  write $q=2s$ and
$(K,M,N)=(k^sm^s,m^sn^s,n^sk^s)$,
 and prune the disjoint MM
  keeping only the term
 $T_s\cdot (K,M,N)=\begin{pmatrix}2s\\s \end{pmatrix}\cdot(k^sm^s,m^sn^s,n^sk^s)$; restrict
 the given decomposition  
 of border rank $r^{2s}$ to this term.
  In this special case we have already proved  Theorem \ref{thsch}, and so we obtain
 $$\omega\le 3\log_R((kmn)^{2s})=
  3\log_{R^{1/2s}}(kmn),~{\rm for }~
 R={r^{2s}/\begin{pmatrix}2s\\s \end{pmatrix}}.$$
Since 
  $$\begin{pmatrix}2s\\s \end{pmatrix}^{(1/2s)}\rightarrow 2~{\rm as}~s\rightarrow \infty$$
  it follows that
   $$R^{1/2s}\rightarrow r/2~{\rm as}~s\rightarrow \infty,~{\rm and~so}~
\omega\le 3\log_{r/2}(kmn),$$ 
which proves the theorem in this special case.

In the general case of $s$ disjoint MM
problems of various sizes $(k_i,m_i,n_i)$, for $i=1,\dots,s$, the same construction and the same proof techniques work. We again perform $q$ recursive steps and arrive at a decomposition of border rank $r^q$ for disjoint MM made up of the terms MM$(K,M,N)$ where
$K=\prod_{i=1}^s k_i^{d_i}$,
$M=\prod_{i=1}^sm_i^{d_i}$,  $N=\prod_{i=1}^sn_i^{d_i}$,  
and $d_1,\dots,d_s$ range 
over all 
$s$-tuples of nonnegative integers 
summed to $q$. 

Group all MM problems of the same size together into the terms of the form 
$T\cdot~{\rm MM}(K,M,N)$ with $T$ denoting 
the coefficients of multinomial expansion.
 Prune the decomposition to each of these terms
 (that is, delete all other terms),  define its APA decomposition of border rank $r^q$.
 Then for every term 
$T\cdot~{\rm MM}(K,M,N)$
we can obtain an upper bound 
 $\omega\le 3 \log_R(KMN)$ for $R=r^q/T$. 
 Maximize these bounds over all such terms, 
 let $q\rightarrow \infty$, and arrive at 
   the claimed bound of Theorem \ref{thsch}
   in the general case.

The name of EXPAND, PRUNE,  and CONQUER
is more adequate for these techniques
than the traditional DIVIDE and CONQUER,
and similarly for the derivation of Bini's bound if we view the interpolation to a single term as pruning. 


\section{Faster Decrease of the Exponent of MM
by Using the EXPAND, PRUNE, and CONQUER
Techniques}\label{sfstdcr}


EXPAND, PRUNE, AND CONQUER techniques
 enabled significant decrease of the exponent $\omega$.
Already the above APA decomposition of ${\rm Trace}(ABD+UVW)$
implies the  upper bound 
$$\omega\le 3\log_{mkn}(0.5(kmn+km+kn))$$ for any triple
of $k$, $m$, and $n$.
Indeed for $k=n=7$ and $m=1$, we arrive
at APA bilinear algorithm of border rank 63 for the pair of disjoint problems of 
MM(7,1,7) and MM(7,7,1) -- of computing  
the outer product of two vectors of dimension 7 and of the product of 
$7\times 7$ matrix by a vector, respectively. Apply Theorem \ref{thsch}
 and obtain 
 \begin{equation}\label{eqdsjapa}
\omega\le 3\log_{49}31.5<2.66.
\end{equation}

Refinement of this construction in  \cite{S81} 
yielded the record estimate
$\omega<2.548$, and \cite{P81} promptly decreased this record bound to $\omega<2.522$
by means of combining APA technique
and trilinear aggregation                                            
 for disjoint MM
 represented by ${\rm Trace}(ABC+UVW+XYZ)$.
Further record of 2.496 was soon established in  
\cite{CW82}.
By continuing this line of research, the  papers \cite{S86}
and \cite{CW90} decreased the upper estimates for the
theoretical exponent of MM  below 2.479 and 2.376, respectively. 

The MM algorithms of these papers 
began with APA decomposition for  disjoint MM of small sizes.
It is instructive to compare the initial trilinear identities in \cite{S86}
and \cite{CW90}
with the decomposition for disjoint MMs 
defined by Tables \ref{tabaggr2}, 
 \ref{tabaggr3}, and \ref{tabaggrapa}.
According to \cite[page 255]{CW90}, ``Strassen used the following trilinear identity, related to $\dots$
trilinear aggregation of \cite{P78}:"
\begin{align*} 
\sum_{i=1}^q\Big (x_0^{[2]}+\lambda x_i^{[1]}\Big )\Big (y_0^{[1]}+\lambda y_i^{[2]}\Big )(z_i
\lambda^{-1}) -x_0^{[2]}y_0^{[1]}\sum_{i=1}^qz_i=\sum_{i=1}^q (x_i^{[1]}y_0^{[1]}z_i+x_0^{[2]}y_i^{[2]}z_i)+O(\lambda).
\end{align*}
This defined a basic APA algorithm of border rank $q+1$ for a pair of block inner products.
 
\cite{CW90} strengthened this construction by 
proposing the following basic APA algorithm of border rank $q+2$ for a triple of block inner products:
\begin{align*}  
\sum_{i=1}^q\lambda^{-2}\Big (x_0^{[0]}+\lambda x_i^{[1]}\Big )\Big (y_0^{[0]}+\lambda y_i^{[1]}\Big )\Big (z_0^{[0]}+\lambda z_i^{[1]}\Big )- ~~~~~~~~~~~~~~~\\  
\lambda^{-3}\Big (x_0^{[0]}+\lambda^2 x_i^{[1]}\Big )\Big (y_0^{[0]}+\lambda^2 y_i^{[1]}\Big )\Big (z_0^{[0]}+\lambda^2 z_i^{[1]}\Big ) +
(\lambda^3-q\lambda^2)x_0^{[0]}y_0^{[0]}z_0^{[0]}=\\
\sum_{i=1}^q (x_0^{[0]}y_i^{[1]}z_i^{[1]}+x_i^{[1]}y_0^{[0]}z_i^{[1]}+x_i^{[1]}y_i^{[1]}z_0^{[0]})+O(\lambda).~~~~~~~~~~~~~~
\end{align*}
 
In both papers \cite{S86} and 
 \cite{CW90} the derivation 
of new record upper bounds on the exponent $\omega$
 from the simple basic designs above required long sophisticated recursive processes and intricate pruning based on amazing and advanced mathematical arguments. Actually 
 paper \cite{CW90}  deduced ``only" the record bound 
$\log_8(4000/27)<2.40364$ from the above design, but then proposed some
extended and more involved 
initial designs and
decreased the bound to the famous barrier  of 2.376. This record
 was only beaten by 0.002 in 2010 \cite{S10} and then by additional 0.001 in 2012--2014 \cite{VW14}. The challenge of reaching the exponent 2 is still open --
 in 2018 the record bound is about 2.3728639 \cite{LG14}. Moreover the study in \cite{AFLG14} showed that  the power of the  extension of the approach of  \cite{CW90}  in the directions of \cite{S10}, \cite{VW14}, and 
 \cite{LG14} is limited, and so  
the decrease of the exponent below 2.37 should require some new dramatically different ideas and techniques. Likewise it was proven in \cite{ASU13} 
 that the group-theoretical approach of \cite{CU03}, \cite{CKSU05} to the acceleration of MM,
 initially considered highly promising for achieving MM in nearly 
 quadratic time,
 must include some new dramatically different ideas and  techniques in order to
 produce any competitive MM algorithm.  
  

\section{Some Impacts of the Study of Fast MM}\label{sfrth1}


The progress in decreasing the exponent $\omega$ towards its lower bound 2 has been essentially in stalemate  for the last three decades, both for feasible MM at the level of about 2.7724 and for unfeasible MM at the level of about 2.38.
 that direction, which has virtually stalled after 1986. In spite of that disappointment
 we believe that overall the study of fast MM was already a success story.
 
\begin{itemize}
\item
Within less than two decades (by 1987)
the straightforward upper bound 3 on
 the MM exponent of
(\ref{eqexpninf}) decreased more than half-way to its lower bound 2 (see  details in Appendix \ref{smmexp}). 
\item
In order to achieve this progress
researchers have found and  revealed new surprising resources and have
developed amazing novel   
techniques,
all built on the top of each other,
involving sophisticated 
combinations of
 trilinear aggregation, 
APA algorithms, disjoint MM,  and
 EXPAND, 
PRUNE, and CONQUER techniques.
\item
The study of fast MM was highly important for the Theory of Computing -- the exponent $\omega$ is one of the most cited quantities in that large field because
progress in its estimation can be immediately extended to a great variety of well-known and intensively studied computational problems, partly listed in Section \ref{sblnpex}.
\item
Besides its impact on the Theory, the progress in decreasing the exponent $\omega$ strengthened the effort of researchers for the reduction of various 
computational problems to MM; such a reduction can be efficient even where the straightforward MM is applied.
\item
Some fast algorithms for feasible MM have been devised, developed, and implemented.  Now they make a valuable part of modern software for both numerical and symbolic computations. Even  limited progress in this direction is valuable because
MM is a fundamental operation of modern computations in Linear and Multilinear Algebra, while
 Polynomial MM makes major impact on  the field of Polynomial Algebra (cf. 
\cite{AHU74}, \cite{BM75}, \cite{BP94}, \cite{BCS97}, \cite{K97}, \cite{GG13}).
\end{itemize}

Next we recall sample by-products of the study of fast MM and of its methodological impact.

\begin{itemize}
\item
Although the origin of the field of fast Algebraic 
Computations can be traced back to \cite{O54}, \cite{M55}, 
 \cite{T55}, and \cite{P66},
the studies of fast MM
 in  \cite{W68},
\cite{P72}, \cite{BD73}, \cite{HM73}, \cite{S73}, \cite{P74},
\cite{P76}, \cite{P78}, and \cite{P84}
as well as APA algorithms in \cite{B80} and
\cite{BCLR79}  
have greatly motivated the effort and the progress in that field.
\item
The duality technique of \cite{P72}
for generating 
new efficient bilinear algorithms,
with  applications shown in \cite{W80} is
a valuable by-product of the MM research.
\item
The MM paper
\cite{P72}
was pioneering in demonstrating the power of the 
application of tensor decomposition to  matrix computations, now a thriving and highly popular area.
\end{itemize}

Finally, in contrast to reasonable pessimism of many experts about current perspectives for further substantial decrease of the exponent $\omega$, the acceleration of feasible MM is a highly promising and dynamic area and, together with the implementation issues, is the main subject of the rest of our survey.


\section{The Curse of Recursion and Fast Feasible MM}\label{scurs}

 
 Already in 1981 it has become clear that
the  progress in decreasing the theoretical  exponent of MM  is going to be
 separated  from 
 the acceleration
of feasible MM. 
Arnold Sch{\"o}nhage
has concluded the introduction to 
his seminal paper \cite{S81} of 1981 as follows:
``Finally it must be stressed, however, that so far all these new results are 
mainly of theoretical interest. The point of intersection with Strassen's
method\footnote{Actually also with the straightforward MM.} lies beyond any practical matrix size, and $\dots$
Pan's estimates of 1978 for moderate values of $n$ are  
still unbeaten".  Sch{\"o}nhage's account can be extended to the subsequent algorithms supporting record estimates for the exponent $\omega$,
except that 
the  estimate 2.7962 of 1978 for the
exponents of feasible MM 
has successively been decreased 
(although by  
small margins) in \cite{P79},
 \cite{P80},  \cite{P81}, and  \cite{P82},  
based purely on trilinear aggregation. 
 By  
2018 the  estimate $w_{44}<2.7734$ of  \cite{P82} 
 still remains the record upper bound on the exponents of  feasible MM,
unbeaten since 1982.  31 years later A.V. Smirnov in \cite{S13} came very close to this record
  by  applying advanced computer-aided search: one of his algorithm supports an exponent below 2.7743
 for $MM(54)$ (see Section \ref{sprspcaid}). 


Tables \ref{tab2} and \ref{tab2a} 
trace the progress in estimating the record  exponents of feasible MM.
The overhead constants associated  with the exponents are reasonably  small
because the
supporting algorithms 
avoid  recursive application of nested block MM and
rely just on trilinear aggregation. This progress in the acceleration of feasible MM 
was much slower than the progress in breaking records for the theoretical  exponent $\omega$, which is no surprise --  the powerful resource of using unlimited recursive processes had to be excluded
for devising  algorithms for feasible MM.  
  
In contrast the techniques of EXPAND, PRUNE, and CONQUER that supports  Bini's and  
Sch{\"o}nhage's theorems as well as the derivation of all known exponents below 2.7733 involve long recursive processes, 
and so  the associated algorithms remain inferior
to the straightforward MM until the problem size is blown up and becomes immense.   Due to such a {\em curse of recursion} 
 all these record breaking works   
 had no relevance to feasible MM
 of today, tomorrow, or  foreseeable future.

\begin{table}[h] 
\caption{Complexity Exponents of Feasible  MM}
\label{tab2}
  \begin{center}
    \begin{tabular}
{| c |c | c| c |c| c|}
      \hline
Exponent  & 2.8074 & 2.7962  & 2.7801 & 2.7762 & 2.7734   \\ \hline
Reference & \cite{S69}   & \cite{P78} & \cite{P80} & \cite{P81} &   \cite{P82}  \\ \hline
Year  &   1969   & 1978 & 1979 & 1981 &   1982  \\ \hline
    \end{tabular}
\end{center}
\end{table}

\begin{table}[h] 
\caption{Ranks and Complexity Exponents of  Feasible MM.}
\label{tab2a}
  \begin{center}
    \begin{tabular}
{ |c | c| c |c| }
      \hline
year  & paper &  rank of $MM(n)$ &  bound on exponent   \\ \hline
1969  &  \cite{S69}   & $n^{2.8074}$ & $\omega_{2}< 2.8074$  \\ \hline
 1978   & \cite{P78} & $(n^3+18n^2-4n)/3$ & $\omega_{70}< 2.7952$   \\ \hline
 1980   & \cite{P80} & $(2n^3+27n^2-2n)/6$ & $\omega_{48}< 2.7802$   \\ \hline
 1981   & \cite{P81} & $(n^3+12n^2+26n)/3$ & $\omega_{46}< 2.7762$   \\ \hline
 1982   & \cite{P82} & $(4n^3+45n^2+ 128n+108)/12$ & $\omega_{44}< 2.7734$   \\ \hline
    \end{tabular}
\end{center}
\end{table}



  

%

All in all, the concept of theoretical exponent of MM has been historically motivated but has not been related to feasible MM.
The complexity exponents of feasible MM have much more relevance to the real world computations,
but in the next section  we significantly  increase the efficiency of feasible MM without decreasing its record complexity exponent. 


\section{Some Ways to Acceleration of Feasible MM with No Decrease of the Complexity Exponent}


\subsection{Acceleration of Complex and Polynomial MM}\label{scmplpl}  


Substitute matrices
for variables of bilinear algorithms
of Examples \ref{excmplpr} and
\ref{excmplpr} 
and obtain efficeint algorithms 
 for Complex and Polynomial MM.

\bigskip

{\bf The 3M Method for Complex MM.}

The rank-3 bilinear algorithm of Example \ref{excmplpr} for multiplying two complex numbers
saves one multiplication but
uses three extra
additions and subtractions in
comparison to the straightforward algorithm of rank 4. Now substitute  $N\times N$ matrices
for the variables $a_1$, $a_2$, $b_1$, and $b_2$ and
 arrive at  the problem 
of multiplying a pair of $N\times N$ 
complex matrices $A_1+{\bf i}A_2$ by $B_1+{\bf i}B_2$.
Then the latter algorithm, called the {\em 3M method},
involves 
$3N^3$ scalar multiplications and 
$3N^3+2N^2$
additions and subtractions versus straightforward $4N^3$
and $4N^3-2N^2$. 
This means saving of about 25\% of all operations 
already for $N=30$ (cf. \cite{H02}).

\bigskip

{\bf Fast Polynomial MM.}

Consider $n\times n$ MM where the input matrices 
are filled with polynomials of degree at most $d-1$.

By applying the straightforward MM
to these input matrices reduce our task to
 performing 
$n^3$ multiplications and $n^3-n^2$
additions of polynomials of degrees at most $d$. All these polynomial additions together involve
just $(n^3-n^2)d$ scalar additions. All polynomial multiplications together involve $(2d^2-2d+1)n^3$
scalar multiplications and additions if 
we apply straightforward polynomial  multiplication, but this bound turns 
into $(4.5K\log_2(K)+2K)n^3$,
for $K=2^k$, $2d-1< K<4d-2$, if 
we apply the FFT-based convolution algorithm of Example \ref{explpr}.

This is dramatic saving if the degree
bound $d$ is  large, but we can save much more  if we  consider the input as two polynomials with matrix coefficients and apply to them fast convolution algorithm of Example \ref{explpr}. In this way we reduce our Polynomial MM to performing at most
$K$ MM$(n)$  and at most $4.5K\log_2(K)+K$ additions, subtractions and  multiplications by scalars of $n\times n$  scalar matrices for $2d-1<K <4d-2$.
With the straightforward MM
we solve these tasks by using at most $2n^3 K+ 
4.5K\log_2(K)n^2$ scalar arithmetic operations overall, which additionally saves for us
$(4.5K\log_2(K)+d)(n-1)n^2$ 
scalar operations versus our first accelerated Polynomial MM.

We can further accelerate both Complex and Polynomial MM (decreasing all our estimates accordingly) if we incorporate fast algorithms for MM
 instead of applying the straightforward MM.


\subsection{Randomized acceleration of feasible MM}\label{srand}

The paper  \cite{DKM06}  presented  surprising
acceleration 
of approximate  rectangular MM
by means of randomization 
 (see a short exposition in \cite[Section 3.1]{M11} 
 and see arXiv 1710.07946 and the bibliography therein for the extension to low rank approximation of a matrix, which a highly popular task linked to fundamental matrix computations and Big Data mining and analysis).
 
We sketch the main result of \cite{DKM06}
by using the spectral and Frobenius matrix norms
$||\cdot||_2$ and $||\cdot||_F$, respectively, in the estimates
for the approximation  errors.

Given two  matrices $A=({\bf a}_j)_{j=1}^n$, of size $m\times n$,
with columns ${\bf a}_j$,
and $B=({\bf b}_j^T)_{j=1}^n$, of size $n\times q$, 
with rows ${\bf b}_j^T$, for $j=1,\dots,n$, 

(i) first compute 
 the so called {\em leverage scores}
(aka {\em importance sampling probabilities}),
$$p_j=\frac{||{\bf a}_j||_2~||{\bf b}_j^T||_2}{\sum_{j'=1}^n||{\bf a}_{j'}||_2~||{\bf b}_{j'}^T||_2}$$
for $j=1,\dots,n$, 

(ii) then randomly select (and re-scale by $1/\sqrt {cp_j}$) $c$ pairs of 
corresponding columns of $A$ and rows of $B$, thereby forming an
$m\times c$ matrix $C$ and a $c\times n$ matrix $R$, and 

(iii) finally compute an approximation $CR$ to the matrix product
$AB$.

It is proven in
\cite{DKM06} that

$$\sqrt c~||CR-AB||_F=O(||A||_F||B||_F)$$
both in expectation and with a high probability.


\subsection{Tensors at Work Again: History and Perspectives
of \\
Com\-puter-Aided Search for Fast MM}\label{sprspcaid} 
 
Properly directed computer-aided search is a natural tool in the search for
efficient basic designs for fast feasible MM, not necessarily directed to the decrease of the complexity exponent.

According to \cite{R79}, computer-aided search has  helped
already in 1979, in the design of 
the APA algorithm of \cite{BCLR79}.
  Even earlier, in  \cite{B70}, Richard P. Brent reduced such a search to a system of nonlinear equations (\ref{eqbrnt}) and proposed to apply least-squares minimization techniques for its solution.  

It is more convenient to use the equivalent system of equations (\ref{eqbrntbr}), whose
solution is in turn   equivalent 
  to finding  the standard canonical decomposition 
 CANDECOMP/PARAFAC for the MM tensor (see, e.g.,  \cite{KB09}). Unfortunately, none of the 
numerous techniques for 3D tensor decomposition has been successful in this 
particular case so far. Loosely speaking, the things are quite different
from the  cases favorable to the known techniques because 
solving Brent's equations requires tensor data expansion rather than compression 
(the latter being the essence of CANDECOMP/PARAFAC technique).
The main reason seems to be rather large lower bounds on the rank of  MM 
in comparison with the rank of the 
general trilinear form $\sum_{i,j,h} t_{i,j,h}a_ib_jd_h$.
 
Now recall the uniqueness theorem of Kruskal  \cite{K77}, which in the case 
of the $MM(n)$ tensor (and with account for the full-rank of the 
three $n^2\times r$ matrices involved) requires that $r<3n^2/2$ for the canonical
decomposition to be essentially unique. This is a clear contradiction 
with the known lower bounds on the  rank $r$ of MM$(n)$ (e.g., $r\ge 2n^2-1$
of \cite{P72}
or $r\ge 3n^2-o(n^2)$ of \cite{S03} and \cite{MR14}).
As a result 
non-unique                                                                                                                                                 solutions of various  kind
 can be observed 
(see, e.g., \cite{B07}). 

Even more destructive for numerical optimization
methods is the presence of infinitely growing approximate solutions, which 
correspond to the existence of APA algorithms. This makes the 
customary tool of applying unconstrained 
ALS optimization inefficient,
 and its modification or alternatives are  required.
(ALS is the acronym for Alternating Least Squares.)

Some success with the ALS method used for 
the minimization of the Euclidean norm of the residuals  of Brent's equations was reported 
in  \cite{JML86} and  \cite{OKM13}. Namely those two studies produced two alternatives to Laderman's 
bilinear algorithm of \cite{L86} for the $MM(3)$ problem with the   
record rank $r=23$.   

  By cleverly extending Brent's approach A.V.
 Smirnov in  \cite{S13}
 achieved important  
 progress.\footnote{He has included 
 \cite{JML86} in his reference list but has not explained that the application of the ALS method to devising fast MM was implicit already in Brent's paper \cite{B70} of  1970.}
He proposed a very special modification of the ALS
procedure based on an adaptive "quantization" of iterated components, 
which luckily resulted in finding an exact solution to $MM(3,3,6)$ with $r=40$. 
This supports the MM complexity exponent $\omega_{54}<2.7743$, 
a remarkably low value for such small matrix sizes, almost matching the record 
$2.7734$ of \cite{P82} for the exponents of feasible MM. 

Unfortunately, the paper \cite{S13} 
 contains neither estimates for the number of scalar additions and subtractions involved in its algorithms nor a constructive recipe for 
the implementation of their 
additive stages. 
In  \cite{BB15} Smirnov's recursive algorithm based 
on $MM(3,3,6)$ has been implemented 
and tested for $MM(n)$
for  dimensions $n$ less than 13,000.
The test results are inferior to
those for  recursive bilinear
processes  based on  Winograd's and Strassen's 
Examples  \ref{ex1} and \ref{ex0},
respectively. 
By no means this comparison is final, however.
The models of communication complexity for serial and parallel computers
are dynamic in time, 
Smirnov 
 and other researchers 
have all chances to strengthen the ALS approach
to fast MM,
producing perhaps 
significant acceleration of practical MM.

Furthermore
the matrices 
  $U$,
$V$, and $W$
of the algorithms of \cite{S13} are  rather densely populated by 
nonzeros (near 50\%)
and may perhaps be sparsified 
if revised algorithms of \cite{S13}  properly incorporate the techniques of 
 TA and disjoint MM.
Indeed success of TA in designing fast  algorithms for feasible MM in \cite{K99} and  \cite{K04} indicates
potential value of that technique 
for simplifying computer-aided search.
In particular trilinear aggregation can exploit the 6-way symmetry in order to reduce the 
 search area for efficient MM algorithms.



\section{Numerical Implementation of Fast MM}\label{simplst}

 
Even the straightforward MM, if it is efficiently implemented, can compete in practice with fast MM. In this section we comment on numerical implementation of that and other efficient algorithms for feasible MM. Implementation of feasible MM must be
efficient in arithmetic cost,
decreasing  vectorization, and data locality (cf. \cite[Chapter 1]{GL13}). We cover all these issues, and notice that
quite frequently all three goals are well compatible 
with each other. 

  
So far in practical numerical computations 
MM is performed
  by means of either 
the straightforward algorithm or 
 recursive bilinear $(2\times 2)$-based MM,
typically the recursive application of
  Winograd's  Example \ref{ex1}
or, more rarely, Strassen's  Example \ref{ex0}.
The implementation of these old
algorithms has been  extensively
worked on by many authors and
makes up a valuable part of the 
present day MM software
(cf.\ \cite{B88}, \cite{H90},
\cite{DHSS94}, \cite{DGP04}, \cite{DGP08}, \cite{BDPZ09},
\cite{DN09}, \cite{BBDLS15},  \cite{BB15},
\cite{BD16}, the references therein, and
in \cite[Chapter 1]{GL13}).
This work,
intensified lately, is still mostly
devoted to the implementation of very old algorithms, ignoring, for example,  
 the advanced implementations of fast MM in
\cite{K99}
and \cite{K04} (see Section \ref{simpta})
 and 
 the significant improvement  
  in \cite{CH17} and \cite{KS17}
  of the recursive bilinear MM based on 
  Examples \ref{ex0} and \ref{ex1}
  by Winograd and  Strassen.
 
We hope that our survey will motivate
 advancing the State of the Art both in the design of fast algorithms for feasible MM and in their efficient implementation.


\subsection {Implementation of Trilinear Aggregation Algorithms}\label{simpta}


\bigskip

Already the first implementations in \cite{K99} and \cite{K04} of the fast algorithms for MM of moderate
sizes based on trilinear aggregation showed their superiority to the alternative approaches regarding
numerical stability, memory consumption, and efficiency for parallel MM.
One can immediately see why so: in the 2- and 3-disjoint product algorithms in the
implementation of \cite{K04}, the coefficient tensors $U$, $V$, and $W$
of (\ref{equvw}) and (\ref{equvw'}) are ``supersparse".

This explains good numerical stability of the latter algorithms according to the customary
measurement by the exponent
2.322 of distinct nature, estimated in 
\cite{K04}. Moreover the paper
\cite{K04} shows significant reduction of the workspace consumption of its
algorithms in comparison with $(2\times 2)$-based MM (that is, recursive bilinear algorithms based on
Winograd's and Strassen's Examples \ref{ex0} and \ref{ex1}, respectively). Namely,
storage consumption decreases from 
$(2/3)N^2$ to $(1/4)N^2$ memory cells
(corresponding to the
MM exponent 2.776), that is, in 8/3 times.
  
The  recent important works \cite{BB15} and \cite{BBDLS15}, both covering the implementation and numerical
stability of fast MM algorithms, have omitted proper discussion of these significant 
benefits for the
implementation of \cite{K04}, apparently leaving the challenge to the study in the future. The papers \cite{K99}
and \cite{K04} had bad luck also with their exposition in the influential paper \cite{DN09}, where the citation ``The
practical implementation of Pan's algorithm $O(n^{2.79}$) is presented by Kaporin [Kaporin 1999; 2004]"
must be corrected into ``The practical implementation of Pan's disjoint matrix product $O(n^{2.811})$-
and $O(n^{2.776})$-algorithms is presented by Kaporin [Kaporin 1999; 2004]".

 Incidentally, the test 
results in the interesting experimental study in \cite{DN09} should be accepted with some caution because
they rely on pseudo random matrices, and such matrices tend to have too good numerical stability.

No further implementation
of trilinear
aggregation algorithms followed so far. This can be understood because  their design is  technically more involved
 and has been much less advertised than
the recursive MM algorithms based on Example \ref{ex1}.

Clearly, further work is in order on the assessment, implementation, amelioration, and extension of such
algorithms for fast feasible MM.


\subsection{Some Imaginary and Real Issues}\label{simplobst}

  
   \begin{enumerate}
 \item
 A considerable group  of numerical
analysts still believes in the folk ``theorem"
that fast MM is always numerical unstable, 
 but in actual tests loss  of accuracy 
in fast MM algorithms
was limited, and formal proofs of
quite reasonable  numerical stability 
of all known fast MM algorithms
is available (see \cite{BL80}, \cite{K99}, \cite{K04},
\cite{DDHK07}, and \cite{DDH07}).
  \item
The paper  \cite{BCD14} emphasizes the importance of non-arithmetic optimization
of matrix algorithms:
``The traditional metric for the efficiency of a numerical algorithm has been the number of arithmetic operations it performs. Technological trends have long been reducing the time to perform an arithmetic operation, so it is no longer the bottleneck in many algorithms; rather, communication, or moving data, is the bottleneck". This statement should be taken not too lightly, but still with a grain of salt: communication cost is limited to operating with the data in primary memory.
 For competent implementations of fast feasible MM 
their arithmetic cost is usually in a rather good accordance with 
 vectorization and communication cost.
 Together (rather than in conflict) with vectorization, numerical stability, and restricting data movement,
arithmetic cost
is still a critical ingredient 
of the evaluation of practical efficiency of MM algorithms (cf. \cite[Chapter 1]{GL13},
 \cite{BB15}
and \cite{BBDLS15}).
For a litmus test, smaller arithmetic cost of Winograd's 
algorithm of Example \ref{ex1} has
made it substantially  more popular 
among the users
than Strassen's algorithm of Example \ref{ex0}
in spite of its a little weaker numerical stability.
\item
Most of nowadays computational methods
are  largely driven by technology. 
In particular presently
computations in single precision are intensively promoted by the manufacturers (mainly NVIDIA) of the GP-GPU (general purpose graphic processing units). In such circumstances, floating-point computations are less preferable compared to the integer residual-based arithmetic, where rational numbers can be used, but using the integers 1, 0, and $-1$ is preferable.
  \end{enumerate}

 
\section{Non-numerical Implementation of Fast MM}\label{simplsymbint}
   
 Implementation of MM in Computer Algebra has become highly efficient when it was reduced to 
 performing MM over {\em word size finite fields}
 with the outputs {\em combined by means of the Chinese Remainder Algorithm}. The paper
 \cite{DGP02} spells out 
 the following principles for computations in finite fields, which were  basic for this success:
 \begin{enumerate} 
\item 
Reduce computations in a finite field  to integer arithmetic with delayed or simultaneous modular reductions;
\item 
 perform integer arithmetic by invoking  floating point units (taking advantage of SIMD
instructions and of numerical BLAS);
\item 
 structure the computations in blocks in order to optimize the use of the memory hierarchy of
the current architectures;
\item 
apply fast MM  algorithms
(so far in practice they are  mostly recursive bilinear algorithm for MM based on
Winograd's $2\times 2$ MM of Example
 \ref{ex1} or less frequently 
based on
Strassen's $2\times 2$ MM of \ref{ex0}, but also Kaporin's algorithms of 
\cite{K99} and \cite{K04} and the 
Any Precision Approximation (APA) algorithm of \cite{BCLR79} are
used).
 \end{enumerate} 
  
 We have two further comments:
 
\begin{itemize} 
\item 
The implementations  
 of fast MM in \cite{K99} and \cite{K04} 
 are particularly attractive within
 this framework because it uses matrices 
 $U$, $V$ and $W$  filled with
  shorter  numbers. Therefore one can 
   perform
  the Chinese Remainder Algorithm 
  for fewer primes and invoke it fewer times.
   \item 
We already pointed out that APA  algorithms are prone to the problems of numerical stability,
but  have good promise for 
   symbolic MM in Computer Algebra
   and for computations with integers (see such important computations  in \cite[Section 8]{P15a}). The current level of the known implementations of APA MM,
    however, is  rudimentary,  staying at the level of the  paper \cite{BCLR79} of 1979,
 and must be moved forward dramatically;
our survey should help to accelerate this process. 
 \end{itemize}
 We refer the reader to \cite{DP16}  
 and the bibliography therein on further details of fast symbolic MM.


\bigskip


{\bf {\LARGE {Appendix}}}
\appendix 


\section{Estimation of the Theoretical Exponent of 
 MM}\label{smmexp}


Tables \ref{tab1} and \ref{tab1a}
show 
 the dynamics of the record estimates for the theoretical exponent of $MM(n)$
 since 1969.  
The tables link each estimate to its recorded
publication in a journal, a conference proceedings, or
as a research report.
It displays the chronological process and reflects
the competition for the decrease of estimates for  the theoretical exponent, 
particularly intensive in 1979--1981 and 1986.     
 
The record upper estimates for the theoretical
exponent  have been 
updated four times during the single year of 1979.
At first the estimate 2.7801 for  appeared in February in a Research Report 
(see \cite{P80}). The estimate 2.7799  
appeared at first as one for an APA-based estimate for the  exponent of MM
in \cite{BCLR79} in June
and then for the theoretical exponent of MM
in \cite{B80}.
 The next upper estimate  2.548 of \cite{S81}  was followed by 
2.522 of \cite{P81}, both published in the   
book of abstracts 
of the conference on the Computational Complexity in Oberwolfach, West Germany, 
organized by Schnorr, Sch{\"o}nhage and Strassen 
(cf. \cite[page 199]{P84} and \cite{S81}). 

The upper bound 2.496 of \cite{CW82}
 was reported in October 1981 at the IEEE FOCS'81,
but in Table \ref{tab1} 
we place it after the estimate 2.517
of the paper \cite{R82} of 1982,
which was submitted in March 1980.
The Research Report version of the paper \cite{CW90}
appeared in August of 1986, 
but in Table \ref{tab1} 
 we place  \cite{CW90} after the paper  
\cite{S86},  published in October of 1986 in the Proceedings of the IEEE FOCS,
because  the paper \cite{S86} has been  
submitted to FOCS'86 in the Spring of 1986 and has been
widely circulated afterwards.

One could complete the historical account
of Tables \ref{tab1} and \ref{tab1a} 
by including our  
 estimate
 2.7804 (announced in
the Fall of 1978, but was
 superseded in February 1979 in \cite{P80})
and the bound 2.5218007, which
 decreased our  estimate 2.5218128 of 1979  
and appeared at the end of
the final version of 
 \cite{S81}  in 1981,\footnote{Observe cross-fertilization:
 Sch{\"o}nhage's disjoint MM and Partial MM have been motivated by trilinear aggregation of \cite{P72} and \cite{P78} and by the design of \cite{BCLR79}, respectively; in turn
the paers \cite{P79} and \cite{P81}  decreased the exponent 2.548 of \cite{S81} to 2.522 by combining his disjoint MM with the technique of  trilinear aggregation.}
that is, before the publication, but after the submission of the  estimate 2.517 of \cite{R82}. 
Table \ref{tab1a} shows the decrease of the record estimates in 1986 -- 2014.
 
We refer the reader 
to  \cite{C82}, \cite{LR83},  \cite{C97},  \cite{HP98}, \cite{KZHP08}, 
  \cite{LG12}, and the bibliography therein
for similar progress in asymptotic acceleration of rectangular MM
and its theoretical implications. 

\begin{table}[h] 
\caption{Record upper estimates for the theoretical exponent of 
 MM.}
\label{tab1}
  \begin{center}
    \begin{tabular}
{| c |c | c| c |c |c| c|}
      \hline
Exponent  & 2.8074 & 2.7962  & 2.7801 & 2.7799 & 2.548 &2.522  \\ \hline
Paper & 
 \cite{S69}   & \cite{P78} & \cite{P80} & \cite{BCLR79}, \cite{B80} & \cite{S81} &\cite{P81} 
\\ \hline
Year & 1969  &  1978  &  1979  &  1979  & 1979 & 1979   \\ \hline
Exponent  & 2.517  & 2.496& 2.479 & 2.376  & 2.374 & 2.373  \\ \hline
Paper  & \cite{R82} & \cite{CW82} &  \cite{S86} & \cite{CW90}  
& \cite{S10}, \cite{DS13} & \cite{VW14}, \cite{LG14}   \\ \hline
Year  & 1980  & 1981  &    1986  &   1986  &  2010  & 2012  \\ \hline
    \end{tabular}
\end{center}
\end{table}

\begin{table}[h] 
\caption{The records in 1986--2014.}
\label{tab1a}
  \begin{center}
    \begin{tabular}
{|c | c| c |c | c| }
      \hline
Exponent\footnote{Actually, the 
proceedings version of  \cite{VW14} in STOC 2012 
claimed a bound
below 2.3727, but this turned out to be premature}
  & 2.3754770 & 2.3736898  
& 2.3729269 & 2.3728639 \\ \hline
Paper &
 \cite{CW90}& 
 \cite{S10}, \cite{DS13} & \cite{VW14} & \cite{LG14} \\ \hline
Year  & 1986 &  2010  & 2012 & 2014 \\ \hline
    \end{tabular}
\end{center}
\end{table}




\section{ FFT, Inverse FFT, and Convolution}\label{scnvfft}


{\em Fast Fourier transform} 
(hereafter referred to as {\em FFT})
 is a celebrated example of  
{\em recursive divide  and conquer algorithms}. It  computes discrete Fourier transform
at $K$ points, 
that is, evaluates a  
polynomial $p(x)=\sum_{i=0}^{K-1}p_ix^i$ 
at the $K$th roots of 1.
For  $K=2^k$ it reduces the task  
to two such problems of half-size: 
$$p(x)=p^{(0)}(y)+xp^{(1)}(y),$$
$$p^{(0)}(y)=\sum_{i=0}^{K/2-1}p_{2i}y^i,~  
~p^{(1)}(y)=\sum_{i=0}^{K/2-1}p_{2i+1}y^i$$ 
where 
$y=x^2$ is a $(K/2)$nd root of 1 if $x$ is a $K$th roots of 1. 
Recursively, the problem 
is reduced  to four problems of quarter size and
ultimately to $K$ problems of size 1,
whose solution ids instant, requiring no
arithmetic operations.
Each of  $k=\log_2(K)$ recursive stages involves
$K$ additions/subtractions
and $K/2$ multiplications ($K/2$ other multiplications are by $-1$ and are 
performed as subtractions).
Thus  the overall  arithmetic cost of FFT is $1.5K\log_2(K)$, versus $K^2$ multiplications and $K^2$ additions
required in the straightforward algorithm. This
is dramatic saving for large $K$, e.g., more than 50,000-fold for $K=1,000,000$. 
 
 Now notice that discrete Fourier transform
computes the vector of values 
${\bf v}$ of a polynomial $p(x)$
at the $K$th roots of unity,
${\bf v}=\Omega {\bf p}$ for the vector
${\bf p}=(p_i)_{i=0}^{K-1}$ of the coefficients of $p(x)$ and the matrix 
$\Omega=(\omega_K^{ij})_{i,j=0}^{K-1}$
filled with  $K$th roots of unity,
where $\omega_K$ is a primitive $K$th root of 1, that is, $\omega_K^K=1$, $\omega_K^h\neq 1$
for $0<h<K$,
\begin{equation}\label{eqrtunt}
\omega_K=\exp(2\pi {\bf i}/K)~{\rm for}~{\bf i}=\sqrt {-1}.
\end{equation}

{\em Inverse discrete Fourier transform}
computes the vector ${\bf p}=\Omega^{-1} {\bf v}$ of the coefficients from a given vector
${\bf v}$ of the values.
It turned out that $\Omega^{-1}=
\frac{1}{K}\Omega^*=\frac{1}{K}(\omega_K^{-ij})_{i,j=0}^{K-1}$,
and one can just readily extend FFT  and then perform 
 $K$ divisions by $K$.

Due to the wide range of important applications of  FFT 
in Modern 
Computations, this recursive divide and conquer  algorithm has become  immensely popular since its publication  
in 1965 in \cite{CT65}
and was
justly included 
into the list of the Ten Top Algorithms of the 20th
century  \cite{C00}, 
even though its origin
has been traced back to posthumous notes of C. F. Gauss, 1777--1855.
See \cite{K97} for early history of FFT and see \cite{AHU74}, \cite{BP94},  \cite[Sections 2.2--2.4]{P01}, 
\cite[Chapter 12. Fast Fourier Transform]{PTVF07}, 
\cite{VL92}, and the bibliography therein
for derivation
of FFT and inverse FFT, their structured matrix 
version, generalization to the case of any integer $K$, 
 numerical stability of the output vector norm, parallel implementation, and further improvements.
  
 
The straightforward algorithm for convolution
involves $(m+1)(n+1)$ multiplications and $mn$ additions,
but the combination of
the Toom's seminal {\em method of evaluation--interpolation} of \cite{T63}
with FFT and Inverse FFT
enables dramatic acceleration.

Let $m=n$ in order to simplify the estimates.
 
\bigskip

{\large \bf Algorithm 1: Convolution via Evaluation, Interpolation, and FFT.}

\bigskip

INITIALIZATION.

 Fix the integer $K=2^k$ being a power of 2 in the range $2n<K\le 4n$,
that is, $k=2+\lfloor\log_2(n)\rfloor$. 

\bigskip

COMPUTATIONS.

\begin{enumerate}
\item
 Compute the values $a(\omega_K^i)$ and  $b(\omega_K^i)$ for $\omega_K$ of (\ref{eqrtunt}) and 
$i=0,\dots,K-1$. They  are the values of the polynomials $a(x)$ and $b(x)$
at the $K$th roots of 1.
\item
Compute the products $c(\omega_K^i)=a(\omega_K^i)b(\omega_K^i)$
for 
$i=0,\dots,K-1$. They  are the values
of the polynomial $w(x)$ at the same $K$ points.
\item
Interpolate to the polynomial $c(x)$ from its values at these points.
\end{enumerate}

Stages 1 and 3 of the algorithm 
 amount to
 multipoint evaluation and interpolation
of polynomials 
at the $K$th roots of 1,  that is,
{\em Forward and Inverse 
Discrete Fourier Transforms},
respectively.  FFT and Inverse FFT
 perform these stages by using 
$K+4.5K\log_2(K)$ 
arithmetic operations. 
Add $K$ bilinear  multiplications,
performed at
Stage 2,
and arrive at the overall  arithmetic computational cost bound $2K+4.5K\log_2(K)$.

 Toom's evaluation--interpolation method  
has a number of further applications to polynomial
and rational computations. For example, it can 
be extended immediately to fast computation of the
quotient $q(x)=u(x)/v(x)$ of two polynomials $u(x)$ and $v(x)$
provided that the remainder of the division is 0;
this restriction has been removed in \cite{PLS92}.








\end{document}